\documentclass[journal,10pt,twocolumn]{IEEEtran}
\usepackage{graphicx}
\usepackage{caption,color}
\usepackage{subcaption}
\usepackage{amsthm,amssymb}
\usepackage{amsmath}
\usepackage{algorithm}
\usepackage{algorithmic}
\usepackage{epstopdf}
\usepackage{mathtools,amssymb}

\newcommand{\rev}[1]{{\color{black}#1}}

\newtheorem{myDef}{Definition}
\newtheorem{myThe}{Theorem}
\newtheorem{myPro}{Proposition}
\begin{document}

\title{\huge Study of Opportunistic Relaying and Jamming Based on Secrecy-Rate Maximization for Buffer-Aided Relay Systems}

\author{Xiaotao~Lu
        and~Rodrigo~C. de~Lamare,~\IEEEmembership{Senior~Member,~IEEE \vspace{-1em}} 
\thanks{Xiaotao Lu is with the Communications Research Group, Department of Electronics, University of York, YO10 5DD York, U.K. e-mail: xtl503@york.ac.uk}
\thanks{R. C. de Lamare is with CETUC, PUC-Rio, Brazil and with the
Communications Research Group, Department of Electronics, University
of York, YO10 5DD York, U.K. email: rcdl500@york.ac.uk.}
}

\maketitle

\begin{abstract}

In this paper, we investigate opportunistic relaying and jamming
techniques and develop relay selection algorithms that maximize the
secrecy rate for multiuser buffer-aided relay networks. We develop
an approach to maximize the secrecy rate of relay systems that does
not require the channel state information (CSI) of the
eavesdroppers. We also devise relaying and jamming function
selection (RJFS) algorithms to select multiple relay nodes as well
as multiple jamming nodes to assist the transmission. In the
proposed RJFS algorithms inter-relay interference cancellation (IC)
is taken into account. IC is first performed to improve the
transmission rate to legitimate users and then inter-relay IC is
applied to amplify the jamming signal to the eavesdroppers and
enhance the secrecy rate. With the buffer-aided relays the jamming
signal can be stored at the relay nodes and a buffer-aided RJFS
(BF-RJFS) algorithm is proposed. Greedy RJFS and BF-RJFS algorithms
are then developed for relay selection with reduced complexity.
Simulation results show that the proposed RJFS and BF-RJFS
algorithms can achieve a higher secrecy rate performance than
previously reported techniques even in the absence of CSI of the
eavesdroppers.

\end{abstract}


%
\IEEEpeerreviewmaketitle

\section{Introduction}

The broadcast nature of wireless communications makes secure
transmissions a very challenging problem. Security techniques
implemented at the network layer rely on encryption keys which are
nearly unbreakable. However, the computational cost of such
encryption algorithms is extremely high. In order to reduce such
cost novel security techniques at the physical layer have been
developed. Physical-layer security was first conceived by Shannon in
his landmark 1949 paper \cite{Shannon} using an information
theoretic viewpoint, where the feasibility of physical-layer
security has been theoretically discussed. Later on, Wyner proposed
a wire-tap channel model that can achieve positive secrecy rates
under the assumption that users have statistically better channels
than those of the eavesdroppers\cite{Wyner}. Since then further
research has been devoted to the wire-tap model in broadcast and
multiple-antenna channels \cite{Csiszar,Oggier,Tie}. Techniques to
enhance the secrecy of wireless systems such as artificial noise
\cite{Goel}, beamforming \cite{Junwei} and relay techniques
\cite{Oohama,wang2017} have also been extensively studied.

\subsection{Previous Work and Problems}

Recently, the concept of physical-layer security with multiuser
wireless networks has been thoroughly investigated and approaches
based on transmit processing and relay techniques have drawn a great
deal of attention \cite{wang2017,Mukherjee,zhao2019,li2019}.
Transmit processing relies on intelligent design of precoding and
signalling strategies to improve the secrecy rate performance. The
use of relays \cite{Dong} and the exploitation of spatial diversity
can also enhance secrecy rates. Moreover, recent advances like
buffer-aided relays have gained significant attention
\cite{Zlatanov,Gaojie,Huang} as they can provide significant
performance advantages over standard relays.

Buffer-aided relay systems with secure constraints have been
investigated in half-duplex \cite{Zlatanov,Gaojie,Shafie1,Huang} and
full-duplex systems \cite{Shafie2}. Opportunistic relay schemes have
been examined with buffer-aided systems in
\cite{Nomikos1,Nomikos2,Lee}. In this context, inter-relay
interference cancellation (IC) at relay nodes is a fundamental
aspect in opportunistic relay schemes. In \cite{Nomikos1}, IC has
been combined with buffer-aided relays and power adjustment to
mitigate inter-relay interference (IRI) and minimize the energy
expenditure. Furthermore, in \cite{Nomikos2} a distributed joint
relay-pair selection has been proposed with the aim of rate
maximization in each time slot using a threshold to avoid increased
relay-pair switching and CSI acquisition. In \cite{Lee} and
\cite{Jingchao}, a jammer selection algorithm and a joint relay and
jammer selection technique have been investigated. The studies in
\cite{Lee} and \cite{Jingchao} have shown that relaying contributes
to a better transmission rate for legitimate users, whereas jamming
can deteriorate the transmission to the eavesdropper. Therefore,
relaying and jamming lead to an improvement in secrecy rate
performance. However, opportunistic buffer-aided relay schemes with
jamming techniques for improving physical layer security have not
been examined so far.

\subsection{Contributions}

In this work, we propose an opportunistic relaying and jamming
scheme and develop relay selection algorithms for the downlink of
multiuser single-antenna and multiple-input multiple-output (MIMO)
buffer-aided relay networks that maximize the secrecy rate, which is
a challenging task due to the difficulty to obtain CSI of the
eavesdroppers. Preliminary results of the proposed techniques have
been reported in \cite{xiaotao2015}, where relaying and jamming
selection have been examined, and in \cite{xiaotao2016}, where relay
selection based on the secrecy rate has been studied. Here, we
devise a relay selection approach for effective secrecy rate (E-SR)
maximization that does not require CSI of the eavesdroppers. The
proposed relaying and jamming function selection (RJFS) algorithms
select multiple relay nodes as well as multiple jamming nodes to
help the transmission. We also present an opportunistic relaying and
jamming scheme in which relaying or jamming is performed within the
same set of relays at different time slots. In the proposed RJFS
algorithms, IC is employed to improve the transmission rate to
legitimate users and the residual interference is used to amplify
the jamming signal to the eavesdroppers. We exploit buffer-aided
relays to store the jamming signals at the relay nodes and devise a
buffer-aided relaying and jamming function selection (BF-RJFS)
algorithm. Greedy RJFS and BF-RJFS algorithms are also developed for
relay selection with reduced complexity. Simulations show that the
proposed RJFS and BF-RJFS algorithms can outperform previously
reported techniques in the absence of CSI of the eavesdroppers. In
addition, the greedy RJFS and BF-RJFS algorithms achieve a
performance close to that of the exhaustive search-based RJFS and
BF-RJFS algorithms, while requiring a much lower computational cost.
The main contributions of this work are:
\begin{itemize}

  \item The E-SR maximization approach that does not require CSI of
  the eavesdroppers is proposed.

  \item An opportunistic relaying and jamming scheme for
  single-antenna and MIMO buffer-aided relay systems.

  \item Novel RJFS algorithms that maximize the secrecy rate are
  developed for buffer-aided relay systems.

  \item Greedy RJFS and BF-RJFS algorithms are developed to reduce the
  computational complexity of exhaustive search-based RJFS and BF-RJFS algorithms.

  \item A secrecy rate analysis of the proposed RJFS algorithms.

\end{itemize}

This paper is organized as follows. In Section II, the system model
and problem formulation are introduced. A review of relay selection
techniques and a novel relay selection criterion without CSI to the
eavesdroppers are included in Section III. The proposed RJFS and
BF-RJFS algorithms are introduced in Section IV. In Section V a
secrecy analysis is carried out. In Section VI, we present and
discuss the simulation results. The conclusions are given in Section
VII.

\subsection{Notation}

{\small { \hspace{-1em} \linespread{1.1}
\begin{tabular}{|l|l|}
  \hline
  Notation & Description \\
  \hline
  \hline
  ${\boldsymbol A}\in {\mathbb{C}}^{M\times N}$ & denotes matrices of size ${M\times N}$ \\
  ${\boldsymbol a}\in {\mathbb{C}}^{M\times 1}$ & denotes column vectors of length $M$ \\
  $(\cdot)^\ast$, $(\cdot)^T$ and $(\cdot)^H$ & represent conjugate, transpose, \\
  & and conjugate transpose, respectively \\
  $\boldsymbol I_{M}$ & is an identity matrix with size $M$ \\
  $\rm diag \{\boldsymbol a\}$ & is a diagonal matrix with the \\
  & elements of $\boldsymbol a$ along its diagonal\\
  $\mathcal{CN}(0,\sigma_{n}^{2})$ & represents complex Gaussian \\
  & random variables with independent  \\
  & and identically distributed ($i.i.d$)  \\
  & entries with mean $0$ and variance $\sigma_{n}^{2}$ \\
  $\log (\cdot)$ & denotes the base-2 logarithm \\
  $\|\boldsymbol A\|_{\rm F}$ & is the Frobenius norm of ${\boldsymbol
  A}$ \\
  ${\boldsymbol H}_{i}\in {\mathbb{C}}^{N_{i}\times N_{t}}$ & is the
  channel matrix from the \\ & transmitter to the $i$th relay\\
  $\boldsymbol L_{\rm state}\in {\mathbb{C}}^{S_{\rm
  total}N_{i}\times L}$ & is the state matrix of the relays \\
  ${\boldsymbol s}^{(t)}\in {\mathbb{C}}^{MN_{i}\times 1}$ & is the
  transmit signal at the source \\
  ${\boldsymbol y}_{i}^{(t)}\in
  {\mathbb{C}}^{N_{i}\times 1}$ and  & refer to the received signals\\
  ${\boldsymbol y}_{r}^{(t)}\in {\mathbb{C}}^{N_{r}\times 1}$& at the
  relays and the destination \\
  $\varGamma_{\rm{IC}-i}^{(t)}$, $\varGamma_{\rm{IC}-e}^{(t)}$ & refer to
  SINR at the $i$th relay node, the\\
  and $\varGamma_{\rm{IC}-r}^{(t)}$ &  $e$th eavesdropper and the
  $r$th receiver\\
  $C_s$ and $R$ & refer to secrecy capacity and rate\\
  $\boldsymbol{\varOmega}^{r}$  &
  refers to the set of $r$ selected relays \\
  $\eta_{\rm Link I}$ and $\eta_{\rm Link II}$ & refer to the selection
  thresholds \\
  $P$ & is the transmit power\\
  \hline
\end{tabular}}}

\section{System Model and Problem Formulation}

In this section, we introduce the multiuser MIMO buffer-aided relay
system model along with details of the proposed opportunistic
relaying and jamming scheme. The physical-layer security problem
associated with the proposed opportunistic relaying and jamming
scheme is then formulated.

\subsection{System Model}
\begin{figure}[h]
\centering
\includegraphics[scale=0.65]{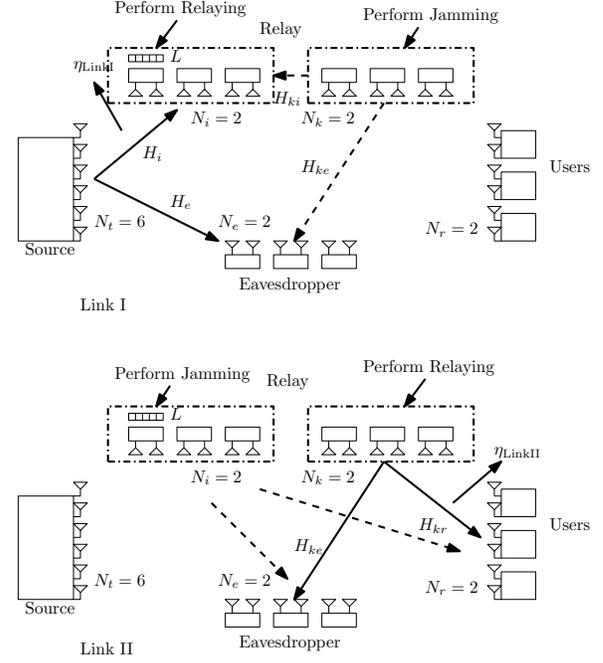}
\vspace{-0.2em}\caption{System model of a multiuser MIMO system with
$M$ users, $N$ eavesdroppers and $S_{\rm total}$ relays.}
\label{fig:sys}
\end{figure}

Fig. \ref{fig:sys} describes the downlink of an opportunistic
multiuser MIMO relay system with $N_{t}$ antennas employed to
transmit data streams aided by precoding to $M$ users in the
presence of $N$ eavesdroppers. The system is equipped with a total
of $S_{\rm total}$ relay nodes and a relay selection scheme that
chooses $S$ out of $S_{\rm total}$ relay nodes. Each relay node is
equipped with $N_i$ antennas and the buffer of each relay can store
$L$ data packets. To show the states of the relays, we use the state
matrix $\boldsymbol L_{\rm state}\in {\mathbb{C}}^{S_{\rm
total}N_{i}\times L}$. Each column of the state matrix $\boldsymbol
L_{\rm state}$ represents the signals stored in the buffers in one
time slot. The buffer state matrix is initialized with zeros.
Similarly to relay systems with two slots, the transmission can be
divided into two parts: Link I and Link II. In Fig. \ref{fig:sys},
the solid lines represent the transmission of intended signals and
the dashed lines denote the transmission of jamming signals. In Link
I, the eavesdroppers try to obtain the signals transmitted from the
source and the jamming signals will cause interference to the
eavesdroppers. Furthermore, in Link II the eavesdroppers attempt to
get the information from the selected relays. In this scenario, the
jamming signals will be generated and transmitted by the relays. In
such opportunistic scheme the relays can be selected to perform
different functions in the same time slot. From Fig. 1 in both Link
I and Link II, although the eavesdropper can accumulate the
information received in both links, the relays selected to perform
jamming will always transmit jamming signals to the eavesdroppers.
Depending on the buffer size, the opportunistic scheme can be
considered in two scenarios:

\begin{itemize}
  \item \textbf{Buffer size $L=1$}: In the first time slot,
  only Link I is employed whereas in the second time slot Link II is used.
  As a result, there is no cooperation between Link I and Link II
  and the temporal advantage is unavailable in this scenario, which means this scheme is
  equivalent to a relay scheme without buffers.

  \item \textbf{Buffer size $L>1$}:
  The thresholds $\eta_{\rm Link I}$ and $\eta_{\rm Link II}$ that indicate the power allocation to the
  transmitter $\eta_{\rm Link I} P$ or $\eta_{\rm Link II}(2-P)$, where $P$ is the power, are calculated separately for Link I and
  Link II, which determines if the relays perform relaying or jamming.
  \begin{itemize}
  \item If $\eta_{\rm Link I}> \eta_{\rm Link II}$, Link I is active. It indicates that the channels from the source to
  the relays can provide a better transmission environment. In this scenario, the jamming signals are generated
  independently at the relays, which are selected to perform the jamming function. The selection of the relays which
  perform the relaying function can be done according to different relay selection criteria.
  The jamming signal will also be stored at the buffers.
   \item If $\eta_{\rm Link I} \leq \eta_{\rm Link II}$, Link II is active. It indicates that the channels
   from the relays to the users have better links. In this scenario, relays will forward
   the signals to the destination. The jamming signals in Link II are the stored jamming signals in Link I, which means
   that the jamming signals in Link II do not need to be generated in Link II.
  \end{itemize}
\end{itemize}

If CSI remains unchanged or one link is always better than the other
than the system will employ a counter, $\eta_{\rm L}$, compare it
with a maximum value $\eta_{L_{\rm max}}$ and activate the link that
has been inactive for $\eta_{L_{\rm max}}$ transmissions. The value
$\eta_{L_{\rm max}}$ is set by the designer. In addition, the change
of links makes it more difficult for the eavesdropper to obtain the
pattern.

In this system, each relay node is equipped with $N_{i}$ antennas.
To indicate when the relays are performing the jamming function, the
relay antenna number is represented by $N_k$. For one relay we can
have $N_k=N_i$. At the receiver side each user and each eavesdropper
is equipped with $N_{r}$ and $N_{e}$ receive antennas. We also
assume that the eavesdroppers do not jam the transmission and the
data transmitted to each user, relay, jammer and eavesdropper
experience a flat-fading MIMO channel. The quantities ${\boldsymbol
H}_{i}\in {\mathbb{C}}^{N_{i}\times N_{t}}$ and ${\boldsymbol
H}_{e}\in {\mathbb{C}}^{N_{e}\times N_{t}}$ denote the channel
matrices of the ith relay and the eth eavesdropper, respectively.
The quantities ${\boldsymbol H}_{ke}\in {\mathbb{C}}^{N_{e}\times
N_{k}}$ and ${\boldsymbol H}_{kr}\in {\mathbb{C}}^{N_{r}\times
N_{k}}$ denote the channel matrices of the kth relay to the eth
eavesdropper and the kth relay to the rth user, respectively. The
channel between the kth relay to the ith relay is represented by
${\boldsymbol H}_{ki}\in {\mathbb{C}}^{N_{i}\times N_{k}}$.

To support the transmission of data to $M$ users, the source is
equipped with $N_{t}\geqslant N_{r}M$ antennas. The total number of
antennas with $S$ relaying function nodes as well as $K$ jamming
function nodes should satisfy $N_{i}S\geqslant N_{r}M$ and
$N_{k}K\geqslant N_{r}M$, respectively. At the same time we assume
that the total number of antennas of the eavesdroppers is $N_{e}N
\geqslant N_{r}M$. In order to satisfy the precoding constraints
\cite{windpassinger}, the number of $N_{r}M$ transmit antennas is
used to transmit signals to $M$ users. The relays can estimate the
channel from the jammers by assuming that there are pilots in the
packet structure, that they know the jamming signals and that the
eavesdroppers cannot decode the jammers. This is reasonable because
the relays also perform jamming and therefore should know the
jamming signals. Moreover, we also assume that CSI of the users can
be obtained at the transmitter by feedback channels from the relays.
Alternatively, advanced parameter estimation and relay techniques
can be employed
\cite{smce,jpais,armo,badstbc,baplnc,jio,jidf,jiols,jiomimo,dce,smaxlink,mwc,chd}.
The vector ${\boldsymbol s}_{i}^{(t)} \in {\mathbb{C}}^{N_{i}\times
1}$ contains the data symbols of each user to be transmitted in time
slot $t$. The transmit signal at the transmitter can be expressed
as:
\begin{equation}
{\boldsymbol s}^{(t)}={\left[ {{\boldsymbol s}_{1}^{(t)}}^{T} \quad
{{\boldsymbol s}_{2}^{(t)}}^{T} \quad  \cdots \quad {{\boldsymbol
s}_M^{(t)}}^{T}\right]}^{T} ~\in {\mathbb{C}}^{MN_{i}\times 1}.
\end{equation}
In previous works \cite{Xiaotao1,Xiaotao_eurasip,Keke2,Keke3,wlbd},
precoding techniques have been applied to mitigate the interference
among users. In this work, we adopt for simplicity linear
zero-forcing precoding whose precoding matrix can be described by
\begin{equation}
\boldsymbol U^{(t)}={\boldsymbol H^{(t)}}^{H}({\boldsymbol
H^{(t)}}{\boldsymbol H^{(t)}}^{H})^{-1}.
\end{equation}
with ${\boldsymbol U}_{i}\in {\mathbb{C}}^{N_{t}\times N_{i}}$,
the total precoding matrix can be expressed as
\begin{equation}
{\boldsymbol U}^{(t)}=\left[ {\boldsymbol U}_{1}^{(t)} \quad
{\boldsymbol U}_{2}^{(t)} \quad \cdots \quad {\boldsymbol
U}_{M}^{(t)}\right],
\end{equation}
and the channel matrix to $S$ selected relays is given by
\begin{equation}
{\boldsymbol H}^{(t)}={\left[ {{\boldsymbol H}_{1}^{(t)}}^{T} \quad
{{\boldsymbol H}_{2}^{(t)}}^{T}  \quad \cdots \quad {{\boldsymbol
H}_{S}^{(t)}}^{T}\right]}^{T}\in {\mathbb{C}}^{SN_{i}\times N_{t}}.
\end{equation}
If the number of antennas equipped at each relay and each user are
the same, the minimum required number of relays is
${\color{black}S=M}$. The channels of the selected relays forwarding
the signals to the $r$th user are described by
\begin{equation}
{\boldsymbol H_{Kr}}^{(t)}={\left[ {{\boldsymbol H}_{1r}^{(t)}}
\quad {{\boldsymbol H}_{2r}^{(t)}} \quad \cdots \quad {{\boldsymbol
H}_{Kr}^{(t)}}\right]}\in {\mathbb{C}}^{N_{r}\times KN_{k}}
\end{equation}
and the channels from the relays to the users are described by
\begin{equation}
{\boldsymbol H}_M^{(t)}={\left[ {{\boldsymbol H}_{K1}^{(t)}}^{T}
\quad {{\boldsymbol H}_{K2}^{(t)}}^{T} \quad \cdots \quad
{{\boldsymbol H}_{KM}^{(t)}}^{T}\right]}^{T}\in
{\mathbb{C}}^{MN_{r}\times KN_{k}}.
\end{equation}
The selected relays also perform jamming for Link I's transmission
to the eavesdroppers, whereas the channels of the jammers to the
$i$th relay are given by
\begin{equation}
{\boldsymbol H_{Ki}}^{(t)}={\left[ {{\boldsymbol H}_{1i}^{(t)}}
\quad {{\boldsymbol H}_{2i}^{(t)}} \quad \cdots \quad {{\boldsymbol
H}_{Ki}^{(t)}}\right]}\in {\mathbb{C}}^{N_{i}\times KN_{k}}.
\end{equation}
In each link, if we assume that the total jamming signals are
${\boldsymbol J}=[{{\boldsymbol j}_{1}}^{T} \quad {{\boldsymbol
j}_{2}}^{T}\quad \cdots \quad {{\boldsymbol j}_{K}}^{T}]^{T}$, the
received signal ${\boldsymbol y}_{i}^{(t)}\in
{\mathbb{C}}^{N_{i}\times 1}$ at each relay node can be expressed by
\begin{equation}
{\boldsymbol y}_{i}^{(t)} = {\boldsymbol H}_{i}\boldsymbol U_{i}{\boldsymbol
s}_{i}^{(t)}+\sum_{j\neq i}{\boldsymbol H}_{i}\boldsymbol U_{j}{\boldsymbol
s}_{j}^{(t)}+{\boldsymbol H}_{Ki}^{(t)}{\boldsymbol
J}+\boldsymbol{n}_{i} \label{eqn:yit}
\end{equation}
In (\ref{eqn:yit}), $\boldsymbol n_i \in \mathcal{CN}(0,\sigma^2_n)$
and the superscript $pt$ designates the previous time slot when the
signal is stored in the buffer at the relay nodes. The quantity
$\sigma^2_n$ is the noise variance for the channel and ${\boldsymbol
H}_{Ki}{\boldsymbol J}$ is regarded as the IRI among the $i$th relay
and the $K$ jammers. The intended relays are selected according to
different criteria, which will be explained later on. The received
signals are expressed by ${\boldsymbol y}^{(pt)}=[{{\boldsymbol
y}_{1}^{(pt_{1})}}^{T} \quad {{\boldsymbol
y}_{2}^{(pt_{2})}}^{T}\quad \cdots \quad {{\boldsymbol
y}_{S}^{(pt_{S})}}^{T}]^{T}$. The superscript $pt$ represents the
time slot and due to the characteristics of buffer relay nodes, the
values can be different for each relay node. According to the
theorem in \cite{Nomikos1}, IRI can be cancelled. The jammers are
targeted towards the $e$th eavesdropper channel described by
\begin{equation}
{\boldsymbol H_{Ke}}^{(t)}={\left[ {{\boldsymbol H}_{1e}^{(t)}}
\quad {{\boldsymbol H}_{2e}^{(t)}} \quad \cdots \quad {{\boldsymbol
H}_{Ke}^{(t)}}\right]}\in {\mathbb{C}}^{N_{e}\times KN_{k}}
\end{equation}
The received signal at the $e$th eavesdropper is then given by
\begin{equation}
{\boldsymbol y}_{e}^{(t)}={\boldsymbol H}_{e}\boldsymbol U_{i}{\boldsymbol
s}_{i}^{(t)} +\sum_{j\neq i}{\boldsymbol H}_{e}\boldsymbol U_{j}{\boldsymbol
s}_{j}^{(t)}+{\boldsymbol H}_{Ke}^{(t)}{\boldsymbol
J}+\boldsymbol{n}_{e}. \label{eqn:yet}
\end{equation}
where $\boldsymbol n_e \in \mathcal{CN}(0,\sigma^2_n)$ is the noise
vector at the eavesdropper. For the eavesdropper, the term
${\boldsymbol H}_{Ke}^{(t)}{\boldsymbol J}$ acts as the jamming
signal, which cannot be removed without CSI knowledge from the kth
jammer to the eth eavesdropper.

If we assume that the transmitted signals from the relays to the
users are expressed as ${\boldsymbol r}^{(t)}$, the received signal
at the destination is given by
\begin{equation}
{\boldsymbol y}_{r}^{(t)}={\boldsymbol H}_{M} {\color{black}{\boldsymbol r}^{(t)}}+\boldsymbol{n}_{r}. \label{eqn:yrt}
\end{equation}
where $\boldsymbol n_r \in \mathcal{CN}(0,\sigma^2_n)$ is the noise
vector at the users.

In the existing IRI scenario based on (\ref{eqn:yit}) when the
transmit signals $\boldsymbol s$ are statistically independent with
unit average energy $\mathbb{E}[\boldsymbol s \boldsymbol
s^{H}]=\boldsymbol I$, the SINR at relay node i
$\varGamma_{\rm{IRI}-i}^{(t)}$ is expressed by
\begin{equation}
\varGamma_{\rm{IRI}-i}^{(t)}=\dfrac{{\gamma}_{S_i,R_{i}}}{\varphi(k,i)
{\gamma}_{R_{K},R_{i}}+{\gamma}_{S_j,R_{i}}+N_{i}}, \label{eqn:SINRiwoc}
\end{equation}
where $\varphi(K,i)$ is the factor that describes the IC feasibility
and $\boldsymbol{\gamma}_{m,n}$ represents the instantaneous
received signal power for the links $m \longrightarrow n$ as
described by
\begin{equation}
{\gamma}_{S_i,R_{i}}=\|\boldsymbol{H}_{i}\boldsymbol{U_i}\|_{\rm F},
\quad {\gamma}_{S_j,R_{i}}=\sum_{j \neq i}\|\boldsymbol{H}_{i}\boldsymbol{U_j}\|_{\rm F},
\label{eqn:gammasr}
\end{equation}
\begin{equation}
{\gamma}_{R_{K},R_{i}}=\|\boldsymbol{H}_{Ki}{\color{black}{\boldsymbol y}^{(pt)}}\|_{\rm F}.
\label{eqn:gammaki}
\end{equation}
The SINR at the $e$th eavesdropper node $\varGamma_{e}^{(t)}$ as
well as the $r$th legitimate user $\varGamma_{r}^{(t)}$ is described
by
\begin{equation}
\varGamma_{\rm{IRI}-e}^{(t)}=\dfrac{{\gamma}_{S_i,E_{e}}}
{{\gamma}_{R_{K},E_{e}}+{\gamma}_{S_j,E_{e}}+N_{e}}, \label{eqn:SINRewoc}
\end{equation}
and
\begin{equation}
\varGamma_{\rm{IRI}-r}^{(t)}=\dfrac{{\gamma}_{R_{k},R_{r}}}{{\gamma}_{R_{K},R_{r}}+N_{r}},
\label{eqn:SINRrwoc}
\end{equation}
where the terms in (\ref{eqn:SINRewoc}) and (\ref{eqn:SINRrwoc}) are given by
\begin{equation}
{\gamma}_{S_i,E_{e}}=\|\boldsymbol{H}_{e}\boldsymbol{U_i}\|_{\rm F},
\quad {\gamma}_{S_j,E_{e}}=\sum_{j \neq i}\|\boldsymbol{H}_{e}\boldsymbol{U_j}\|_{\rm F},
\label{eqn:gammase}
\end{equation}
\begin{equation}
{\gamma}_{R_{K},E_{e}}=\|\boldsymbol{H}_{Ke}{\boldsymbol J}\|_{\rm F}
\label{eqn:gammare}
\end{equation}
and
\begin{equation}
{\gamma}_{R_{k},R_{r}}=\|\boldsymbol{H}_{kr}
\boldsymbol{J}_{k}\|_{\rm F}.
\label{eqn:gammakr1}
\end{equation}
\begin{equation}
{\gamma}_{R_{K},R_{r}}=\|\boldsymbol{H}_{Kr}
\boldsymbol{J}-\boldsymbol{H}_{kr}
\boldsymbol{J}_{k}\|_{\rm F}.
\label{eqn:gammakr}
\end{equation}
Depending on the IRI cancelation (IC) at the relay nodes, two type
of schemes can be applied. According to \cite{Nomikos1}, if we
assume $\mathbb{E}[\boldsymbol s \boldsymbol s^{H}]=\boldsymbol I$
and $\mathbb{E}[{\color{black}{\boldsymbol y}^{(pt)} {{\boldsymbol
y}^{(pt)}}^{H}}]=\boldsymbol I$, the feasibility of IC can be
described by a factor $\varphi(K,i)$ which is described by
\begin{equation}
\varphi(K,i)=
\begin{cases}
0& \text{if $\det \Big(({\boldsymbol{H}_{i} \boldsymbol{H}_{i}^{H}
+\boldsymbol{I}})^{-1}{\boldsymbol{H}_{Ki}\boldsymbol{H}_{Ki}^{H}}\Big)\geqslant \gamma_{0}$}\\
1 &\text{otherwise},
\end{cases}
\label{eqn:varphi}
\end{equation}
where $\varphi(K,i)=0$ means the interference can be cancelled from
the received signal at the relays, whereas $\varphi(K,i)=1$ means IC
should not be performed. The quantity $\gamma_{0}$ is the threshold
that indicates the feasibility of IC, which is obtained by
simulation. We assume that the channels from the relays, which
perform jamming, are available at the transmitter. 

In the IC scenario, interference mitigation can be performed at the
relay nodes by setting $\varphi(K,i)=0$. The SINR expressions at the
$i$th relay node, the $e$th eavesdropper and the $r$th receiver are
respectively given by
\begin{equation}
\varGamma_{\rm{IC}-i}^{(t)}=\dfrac{{\gamma}_{S_i,R_{i}}}{{\gamma}_{S_j,R_{i}}+N_{i}},\quad
\varGamma_{\rm{IC}-e}^{(t)}=\dfrac{{\gamma}_{S_i,E_{e}}}{{\gamma}_{R_{K},E_{e}}+{\gamma}_{S_j,E_{e}}+N_{e}}
\label{eqn:SINRiwc}
\end{equation}
and
\begin{equation}
\varGamma_{\rm{IC}-r}^{(t)}=\dfrac{{\gamma}_{R_{k},R_{r}}}{{\gamma}_{R_{K},R_{r}}+N_{r}}.
\label{eqn:SINRrwc}
\end{equation}
Alternative interference mitigation techniques can also be
considered
\cite{mmimo,lsmimo,spa,mfsic,mbdf,did,rrser,bfidd,1bitidd,aaidd,aaidd,listmtc}
\cite{jidf,jidf_echo,sjidf,rccm,rrdoa,okspme,rrsgp,rrser,l1stap,rdrcb,ccg,locsme,rrstap,l1stap2,jiomimo,smtvb}
\cite{jiocdma,aifir,intadap,inttvt,jio,jiols,spa,mbdf,dfjio,jioccm,rcb,rhomo,dcdrec,ccmmwf,wlmwf,wlbeam,barc,saabf,damdc,kaesprit,dce,saalt}
\cite{siprec,gbd,wlbd,mbthp,rmbthp,bbprec,baplnc,memd,locsme,okspme,lrcc}

\subsection{Problem formulation}

In this subsection, we describe the secrecy rate used in the
literature to assess the performance of the proposed algorithms in
physical-layer security systems and formulate the problem. The MIMO
system secrecy capacity \cite{Oggier} is expressed by
\begin{equation}
\begin{split}
C_{s} &=\max_{{\boldsymbol Q}_{s}\geq 0, \rm Tr({\boldsymbol Q}_{s})= E_{s}}\log(\det({\boldsymbol I}+{\boldsymbol H}_{ba} {\boldsymbol Q}_{s} {\boldsymbol H}_{ba}^H))\\
& \quad -\log(\det({\boldsymbol I}+ {\boldsymbol H}_{ea}{\boldsymbol Q}_{s} {\boldsymbol H}_{ea}^H)),
\end{split}
\label{eqn:Rs1}
\end{equation}
where ${\boldsymbol Q}_{s}$ is the covariance matrix associated with
the signal and ${\boldsymbol H}_{ba}$ and ${\boldsymbol H}_{ea}$
represent the links between the source to the users and the
eavesdroppers, respectively. For relay systems \cite{Lee}, according
to (\ref{eqn:yit}) and (\ref{eqn:yrt}), with equal power $P$
allocated to the transmitter and the relays, the achievable rate of
the users is given by
\begin{equation}
R_{r}= \log(\det({\boldsymbol I}+\boldsymbol{\Gamma}_{r}^{(t)}  ) )
\label{eqn:RR1}
\end{equation}
where $\boldsymbol{\Gamma}_{r}^{(t)}$ according to
(\ref{eqn:SINRrwoc}) is given by
\begin{equation}
\boldsymbol{\Gamma}_{r}^{(t)} = \frac{P}{\sum_{k=1}^{K}{N_k}}{\boldsymbol
H}_{Kr} {\boldsymbol H}_{Kr}^H({\boldsymbol I}+\frac{P}{N_t}{\boldsymbol
H}^{(pt)}{\boldsymbol U^{pt}}{\boldsymbol U^{pt}}^{H} {{\boldsymbol H}^{(pt)}}^H ). \label{eqn:RR2}
\end{equation}
Similarly, the achievable rate of eavesdroppers is described by
\begin{equation}
R_{e}=\log(\det({\boldsymbol I}+\boldsymbol{\Gamma}_{e}^{(t)}  ) )
\label{eqn:RE1}
\end{equation}
and the $\boldsymbol{\Gamma}_{e}^{(t)}$ according to (\ref{eqn:SINRewoc}) is described by
\begin{equation}
\boldsymbol{\Gamma}_{e}^{(t)} = ({\boldsymbol I
+\boldsymbol\varDelta})^{-1}{\frac{P}{N_{t}}{\boldsymbol
H}_{e}\boldsymbol U \boldsymbol U^{H}{{\boldsymbol H}_{e}}^H}, \label{eqn:RE2}
\end{equation}
where
\begin{equation}
\boldsymbol\varDelta = {
\sum_{e=1}^{N}\frac{P}{\sum_{k=1}^{K}N_{k}}{\boldsymbol H}_{Ke} {\boldsymbol
H}_{Ke}^H({\boldsymbol I}+\frac{P}{N_{t}}{\boldsymbol
H}^{(pt)}{\boldsymbol U^{pt}}{\boldsymbol U^{pt}}^{H} {{\boldsymbol H}^{(pt)}}^H)}. \label{eqn:RE21}
\end{equation}
In (\ref{eqn:RE2}), $\boldsymbol\varDelta$ is the jamming signal to
the eavesdropper. Using (\ref{eqn:RR1}) and (\ref{eqn:RE1}) the
secrecy rate is given by
\begin{equation}
R=\sum_{r=1}^{S}\sum_{e=1}^{N}[R_{r}-R_{e}]^{+} \label{eqn:R}
\end{equation}
where $[x]^{+}=\max(0,x)$. In (\ref{eqn:R}), we assume that each
eavesdropper will listen to the information transmitted to a
particular user. However, the assumption of the availability of
global CSI knowledge is impractical, especially for the
eavesdroppers. For this reason, we consider partial CSI knowledge to
the relays as well as to the users. The problem we are interested in
solving is to select the set of relay nodes to perform relaying or
jamming based on the maximization of the secrecy rate. Therefore,
the proposed optimization problem can be formulated as:
\begin{equation}
\begin{aligned}
& \underset{\boldsymbol{\varOmega}^{r},\boldsymbol{\varOmega}^{m}}{\text{maximize}}
& & R \\
& \text{subject ~to} & &  \boldsymbol{\varOmega}^{r}, \boldsymbol{\varOmega}^{m} \in \boldsymbol{\Psi}  \\
\end{aligned}
\label{eqn:RF}
\end{equation}
where $\boldsymbol{\Psi}$ represents the collection of relay subsets
and $\boldsymbol{\varOmega}^{r}$ and $\boldsymbol{\varOmega}^{m}$
denote the set of selected jamming function nodes and the set of
relaying function nodes, respectively.

\section{Relay Selection Algorithms}

In conventional relaying or jamming systems, relays always perform
as the transmitter and the receiver to enhance the signal
transmission from the source to the destination \cite{tds}. We first
review several algorithmic solutions under this conventional relay
scenario and then present the proposed relay selection based on the
secrecy rate, which does not require knowledge of CSI to the
eavesdroppers.

\subsection{Conventional Relay Selection}

Conventional relay selection does not take the jamming function of
relay nodes into account and the relay nodes are selected with
different selection criteria to assist the transmission between the
source and the destination with only one eavesdropper \cite{Gaojie}
or without consideration of eavesdroppers \cite{Jingchao,Krikidis}.

In \cite{Krikidis}, a max-min relay selection has been considered as
the optimal selection scheme for conventional decode-and-forward
(DF) relay setups. In a single-antenna scenario the relay selection
is given by
\begin{equation}
R_{i}^{\rm max-min}=\rm{arg} \max_{R_{i} \in \boldsymbol{\Psi}} \min(\|h_{S,R_{i}}\|^{2},\|h_{R_{i},D}\|^{2})
\label{eqn:maxm}
\end{equation}
where $h_{S,R_{k}}$ is the channel gain between the source and the
relay and $h_{R_{k},D}$ is the channel gain between the relay $k$
and the destination. Similarly, a max-link approach has also been
introduced to relax the limitation that the source and the relay
transmission must be fixed. The max-link relay selection strategy
can be described by
\begin{equation}
\begin{split}
R_{i}^{\rm max-link}=&\rm{arg} \max_{R_{i} \in \boldsymbol{\Psi}}\big( \bigcup_{R_{i} \in
\zeta:\varphi{(Q_{p})}\neq L}\|h_{S,R_{i}}\|^{2},\\
&\bigcup_{R_{i} \in
\boldsymbol{\Psi}:\varphi{(Q_{p})}\neq 0}\|h_{R_{i},D}\|^{2}\big) \label{eqn:maxl}
\end{split}
\end{equation}
With the consideration of the eavesdropper, a max-ratio selection policy is
proposed in \cite{Gaojie} and is expressed by
\begin{equation}
R_{i}^{\rm max-ratio}=\rm{arg} \max_{R_{i}  \in \boldsymbol{\Psi}}\left(
\eta_{1},\eta_{2}\right) \label{eqn:maxr}
\end{equation}
with
\begin{equation}
\eta_{1}=\frac{\max_{R_{i} \in \boldsymbol{\Psi}:\varphi{(Q_{p})}\neq
L}\|h_{S,R_{i}}\|^{2}}{\|h_{se}\|^{2}}
\label{eqn:eta1}
\end{equation}
\begin{equation}
\eta_{2}=\max_{R_{i} \in \boldsymbol{\Psi}:\varphi{(Q_{p})}\neq
0}\frac{\|h_{R_{i},D}\|^{2}}{\|h_{R_{ie}}\|^{2}}
\label{eqn:eta2}
\end{equation}
The aforementioned relay selection procedure is based on knowledge
of CSI. 

\subsection{Optimal Selection (OS)}

Since conventional relay selection \cite{Jingchao} may not support
systems with secrecy constraints, we consider optimal selection (OS)
which takes the eavesdropper into consideration. The SINR of OS in
the downlink of multiuser MIMO relay systems under consideration can
be expressed similarly to (\ref{eqn:SINRewoc}) and
(\ref{eqn:SINRrwoc}), as described by
\begin{equation}
\varGamma_{e}^{(t)}=\dfrac{{\gamma}_{S_i,E_{e}}}{{\gamma}_{S_j,E_{e}}+N_{e}}
\label{eqn:SINRewcos}
\end{equation}
and
\begin{equation}
\varGamma_{r}^{(t)}=\dfrac{{\gamma}_{R_{k},R_{r}}}{{\gamma}_{R_{K},R_{r}}+N_{r}}.
\label{eqn:SINRrwcos}
\end{equation}
The OS algorithm is given by
\begin{equation}
\begin{split}
R^{OS}&=\rm{arg} \max {[R_{r}-R_{e}]^{+}}\\
     &=\rm{arg} \max {[\log(1+ \varGamma_{r}^{(t)} ) )-\log(1+\varGamma_{e}^{(t)} )]^{+}}\\
\end{split}
\label{eqn:ROS}
\end{equation}

{\subsection{Proposed Effective Secrecy-Rate Relay Selection}

In the previously described relay selection algorithms, the
availability of CSI to the eavesdroppers is an adopted assumption in
the design of relay selection algorithms with secrecy constraints.
However, in the optimization problem in (\ref{eqn:RF}), the CSI of
the eavesdroppers is not available to the transmitter and the users.
In order to circumvent this limitation, we propose a novel relay
selection criterion that is termed effective secrecy rate (E-SR),
which does not require CSI to the eavesdroppers and is incorporated
in the multiuser MIMO buffer-aided relay system under study. The
proposed E-SR approach is based on the maximization of the secrecy
rate and introduces a simplification in the computation of the
expression that does not require the knowledge of CSI to the
eavesdroppers. The proposed E-SR approach for selecting multiple
relays is expressed by {\small
\begin{equation}
\begin{split}
\hspace{-0.5em}{\mathcal R}^{\rm S-SR} & = \arg \max_{\boldsymbol
\varphi \in \boldsymbol{\Psi}}\sum_{i \in \boldsymbol \varphi}
\bigg\{\log\big(\det{\left[\boldsymbol I+(\boldsymbol
H_{i}\boldsymbol R_{I}\boldsymbol H_{i}^{H})^{-1}(\boldsymbol
H_{i}\boldsymbol R_{d}\boldsymbol H_{i}^{H})\right]}\big) \\
& \qquad -\log\big(\det{\left[\boldsymbol I + \boldsymbol
U_{i}^{H}\boldsymbol R_{I}^{-1}\boldsymbol U_{i}\boldsymbol
R_{d}\right]}\big)\bigg\}, \label{eqn:alrnew}
\end{split}
\end{equation}}
where the covariance matrix of the interference and the signal can
be described as $\boldsymbol R_{I}=(\boldsymbol
H_i)^{-1}(\boldsymbol H_i^H)^{-1}+\sum_{j\neq i}\boldsymbol
U_{j}{\boldsymbol s}_{j}^{(t)}{{\boldsymbol
s}_{j}^{(t)}}^{H}{\boldsymbol U_{j}}^{H}$ and $\boldsymbol
R_{d}=\boldsymbol U_{i}{\boldsymbol s}_{i}^{(t)}{{\boldsymbol
s}_{i}^{(t)}}^{H}{\boldsymbol U_{i}}^{H}$, respectively. The details
of E-SR relay selection criterion are given in the Appendix. In
(\ref{eqn:alrnew}), no CSI to the eavesdroppers is required and the
E-SR approach only depends on the CSI to the intended receiver and
the covariance matrix of the interference and the signal. In the
following proposed relaying and jamming schemes, the E-SR technique
is applied to circumvent the need for global instantaneous CSI of
the eavesdroppers.

\section{Relaying and Jamming Function Selection Algorithms}

In this section, we detail the proposed RJFS and BF-RJFS algorithms
along with their cost-effective greedy versions for single-antenna
and multiple-antenna scenarios.

\subsection{Relaying and Jamming Function Selection (RJFS)}

We assume that the total number of relay nodes is $S_{\rm total}$
and $\boldsymbol \Omega$ is the total relay set. To apply the
opportunistic scheme in the system, an initial state is set
according to the channel:
\begin{equation}
\boldsymbol \varOmega^{0,*}=\rm{arg} \max_{\boldsymbol \varOmega^m }
\det\Big({\boldsymbol{H}_{\boldsymbol \varOmega^m}
{\boldsymbol{H}_{\boldsymbol \varOmega^m}}^{H}}\Big), \label{eqn:varGammam}
\end{equation}
where ${\boldsymbol{H}}_{\boldsymbol \varOmega^{m}}$ refers to the
set of channels examined prior to selection and we assume that in
the initial state the relays will not perform the jamming function.
$S$ relay nodes are selected according to the criterion as explained
in Section II. With the total number of relaying and jamming nodes
$S_{\rm total}$ and the number of selected nodes in each group $S$,
the selection operation can be expressed as:
\begin{equation}
\boldsymbol \Psi=\binom{{\color{black}S_{\rm total}}}{{\color{black}S}}, \label{eqn:Omega}
\end{equation}
where $\boldsymbol \Psi$ represents the total number of sets of S
combinations and in each set there are $S$ selected relaying or
jamming nodes. For a particular set $\boldsymbol \varOmega^{m}$, the
channel matrix of selected sets can be described by
\begin{equation}
\boldsymbol H_{\boldsymbol \varOmega^{m}} = {\left[ {{\boldsymbol
H}_{\boldsymbol \varOmega_{1}^{m}}^{(t)}}^{T} \quad {{\boldsymbol
H}_{\boldsymbol \varOmega_{2}^{m}}^{(t)}}^{T} \quad \cdots \quad
{{\boldsymbol H}_{\boldsymbol
\varOmega_{\color{black}S}^{m}}^{(t)}}^{T}\right]}^{T}.
\label{eqn:Omega2}
\end{equation}
If the total collection of selected sets is represented by
$\boldsymbol{\Psi}_{\rm Relaying}$, then for each set the relay
selection is given by
\begin{equation}
\begin{split}
\boldsymbol{\varOmega}^{m,*}& = \rm{arg}
\max_{\boldsymbol{\varOmega}^{m} \in \boldsymbol{\Psi}_{\rm
Relaying}}\sum
\bigg\{\log\big(\det{(\boldsymbol{\Gamma}_{\boldsymbol
\varOmega^{m}}^{(t)})}\big)
\\ & \quad -\log\big(\det{(\boldsymbol{\Gamma}_{\boldsymbol
\varOmega^{e}}^{(t)})}\big)\bigg\} \label{eqn:Omega3}
\end{split}
\end{equation}
where
\begin{equation}
\boldsymbol{\Gamma}_{\boldsymbol \varOmega^{m}}^{(t)} =\boldsymbol
I+ (\boldsymbol
H_{i}\boldsymbol R_{I}^{\boldsymbol \varOmega^{m}}\boldsymbol
H_{i}^{H})^{-1}(\boldsymbol
H_{i}\boldsymbol R_{d}^{\boldsymbol
\varOmega^{m}}\boldsymbol H_{i}^{H})
\label{eqn:Omega4}
\end{equation}
and
\begin{equation}
\boldsymbol{\Gamma}_{\boldsymbol \varOmega^{e}}^{(t)}=\boldsymbol
I+\boldsymbol U_{i}^{H} {\boldsymbol R_{I}^{\boldsymbol
\varOmega^{m}}}^{-1}\boldsymbol U_{i}\boldsymbol R_{d}^{\boldsymbol
\varOmega^{m}} \label{eqn:Omega8}
\end{equation}
In (\ref{eqn:Omega4}) and (\ref{eqn:Omega8}), the covariance
matrices $\boldsymbol R_{I}^{\varOmega^{m}}$ and $\boldsymbol
R_{d}^{\varOmega^{m}}$ can be obtained in the same way as
illustrated in (\ref{eqn:alrnew}). The only difference of the RJFS
algorithm resides in the calculation of $\boldsymbol
R_{I}^{\varOmega^{m}}$, apart from the interference from different
users, there is also existing interference from the jamming function
relay nodes. With the same distributions of the channels from the
jamming function relay nodes to the eavesdropper, $\boldsymbol
R_{I}^{\varOmega^{m}}$ can be calculated in a similar way to that in
(\ref{eqn:alrnew}). In Algorithm~\ref{alg:rjfs} the main steps of
RJFS are given. Step 1 of Algorithm ~\ref{alg:rjfs} gives the
collection of relay subsets, which contain the combinations of $S$
relay nodes out of $S_{\rm total}$ relay nodes. In our definition,
$\boldsymbol \Psi_{\rm Relaying}$ is the same as $\boldsymbol \Psi$.
However, to indicate the differences in the buffer relay system, we
use $\boldsymbol \Psi_{\rm Relaying}$ instead of $\boldsymbol \Psi$.
Note that in both RJFS and BF-RJFS algorithms, we use $\boldsymbol
\Psi_{\rm Relaying}$ as a collection of relay subsets which perform
the relaying function. With no buffers implemented at the relay
nodes, the RJFS algorithm only selects the relays in Link I. The
relays used in Link II are the same as those selected in Link I.

\begin{algorithm}
\caption{RJFS Algorithm}\label{alg:rjfs}
\begin{algorithmic}[1]
\REQUIRE $\boldsymbol H_i$, $\boldsymbol R_I$ $\boldsymbol R_d$,
{$S_{\rm total}$} and {$S$}

\STATE $\boldsymbol
\Psi_{\rm Relaying}=\binom{{S_{\rm total}}}{{S}}$

\COMMENT{ Select $S$ relay nodes out of $S_{\rm total}$ nodes, all
combinations are stored in $\boldsymbol \Psi_{\rm Relaying}$} \STATE
$[\varOmega_c \quad \varOmega_r]={\rm size}(\boldsymbol \Psi_{\rm
Relaying})$

\COMMENT{ Give the matrix size of $\boldsymbol \Psi_{\rm Relaying}$}
\FOR {$m=1:\varOmega_c$}
\STATE $\boldsymbol{\Gamma}_{\boldsymbol \varOmega^{m}}^{(t)} =\boldsymbol
I+ (\boldsymbol
H_{i}\boldsymbol R_{I}^{\boldsymbol \varOmega^{m}}\boldsymbol
H_{i}^{H})^{-1}(\boldsymbol
H_{i}\boldsymbol R_{d}^{\boldsymbol
\varOmega^{m}}\boldsymbol H_{i}^{H})$ \STATE
$\boldsymbol{\Gamma}_{\boldsymbol \varOmega^{e}}^{(t)}=\boldsymbol I+\boldsymbol U_{i}^{H}{\boldsymbol R_{I}^{\boldsymbol \varOmega^{m}}}^{-1}\boldsymbol U_{i}\boldsymbol R_{d}^{\boldsymbol \varOmega^{m}}
$ \STATE $\boldsymbol\varGamma(\boldsymbol \varOmega^{m})=\sum
\bigg\{\log\big(\det{(\boldsymbol{\Gamma}_{\boldsymbol \varOmega^{m}}^{(t)})}\big) -\log\big(\det{(\boldsymbol{\Gamma}_{\boldsymbol \varOmega^{e}}^{(t)})}\big)\bigg\}
$

\COMMENT{Calculate the threshold $\boldsymbol\varGamma(\boldsymbol \varOmega^{m})$ for the $\boldsymbol \varOmega^{m}$ combination}
\ENDFOR
\STATE ${\boldsymbol{\varOmega}^{m,*}}=\rm{arg} \max_{\boldsymbol{\varOmega}^{m} \in \boldsymbol{\Psi}_{\rm Relaying}}(\boldsymbol\varGamma(\boldsymbol \varOmega^{m}))$

\COMMENT{Obtain the combination which gives the optimal value}
\RETURN The set of the selected relays ${\boldsymbol{\varOmega}^{m,*}}$

\end{algorithmic}
\end{algorithm}

\subsection{Buffer-Aided Relaying and Jamming Function Selection (BF-RJFS)}

Here we describe the proposed BF-RJFS algorithm, which exploits
relays equipped with buffers. Based on the RJFS algorithm, the
selection of the $S$ relays used for signal reception is the same as
that in the buffer relay scenario. The main difference between the
proposed BF-RJFS and RJFS algorithms relies on the selection of the
jammer. The selection of the set of jamming and communication relays
is performed simultaneously. According to (\ref{eqn:Omega3}), we
assume the corresponding threshold is stored in
$\boldsymbol\varGamma$. Given the total collection of jamming
selections $\boldsymbol{\Psi}_{\rm Jamming}$, the remaining relays
are selected according to the proposed E-SR criterion as described
by
\begin{equation}
\begin{split}
\boldsymbol{\varOmega}^{r,*} & = \rm{arg}
\max_{\boldsymbol{\varOmega}^{r} \in \boldsymbol{\Psi}_{\rm
Jamming}}\sum \bigg\{\log\big(\det{(\boldsymbol{\Gamma}_{\boldsymbol
\varOmega^{r,n}}^{(t)})}\big) \\ & \quad
-\log\big(\det{(\boldsymbol{\Gamma}_{\boldsymbol
\varOmega^{r,e}}^{(t)})}\big)\bigg\}, \label{eqn:sinrct2}
\end{split}
\end{equation}
where $\boldsymbol{\Gamma}_{\boldsymbol \varOmega^{r,n}}^{(t)}$ is given by
\begin{equation}
\boldsymbol{\Gamma}_{\boldsymbol \varOmega^{r,n}}^{(t)}= \boldsymbol I +(\boldsymbol H_{\boldsymbol
\varOmega^{r}}\boldsymbol R_{I}^{BF}\boldsymbol H_{\boldsymbol \varOmega^{r}}^{H})^{-1}(\boldsymbol H_{\boldsymbol
\varOmega^{r}}\boldsymbol R_{d}^{BF}\boldsymbol H_{\boldsymbol \varOmega^{r}}^{H}),
\label{eqn:gamman}
\end{equation}
where $\boldsymbol R_{d}^{BF}$ is the covariance matrix of the
transmit signal from the jamming function relay nodes to the users.
The jamming signal is the same as the received signal from the
relays in previous time slots. The calculation of $\boldsymbol
R_{d}^{BF}$ depends on (\ref{eqn:gammakr}) and $\boldsymbol
R_{I}^{BF}$ relies on (\ref{eqn:gammase}) and (\ref{eqn:gammare}).
In this procedure, the calculation of
$\boldsymbol{\Gamma}_{\boldsymbol \varOmega^{r,n}}^{(t)}$ is
obtained by
\begin{equation}
\boldsymbol{\Gamma}_{\boldsymbol \varOmega^{r,e}}^{(t)}=\boldsymbol
I+ \boldsymbol U_{r}^{H}{\boldsymbol R_{I}^{BF}}^{-1}\boldsymbol
U_{r}\boldsymbol R_{d}^{BF}, \label{eqn:gammae2}
\end{equation}
where the relays used for jamming in the next time slot are
selected. With the selection of communication relays and jamming
relays the system can provide a better secrecy performance as
compared to conventional relay systems. In
Algorithm~\ref{alg:bfrjfs} the main steps of BF-RJFS are outlined.
Steps 1 to 8 of Algorithm~\ref{alg:bfrjfs} eliminate the relay nodes
with empty buffers because they cannot perform relaying function.
Steps 9 to 20 eliminate relay nodes with a full buffer as the
signals from the source cannot be stored in these relay nodes. Steps
21 to 25 return the results of the subset of selected relay nodes.

{\footnotesize
\begin{algorithm}
\caption{BF-RJFS Algorithm}
\label{alg:bfrjfs}
\begin{algorithmic}[1]
{\REQUIRE $\boldsymbol H_i$, $\boldsymbol R_I$,  $\boldsymbol R_d$,
$\boldsymbol H_{\boldsymbol \Omega^r}$, precoding matrix
$\boldsymbol U_r$, $\boldsymbol R_I^{BF}$, $\boldsymbol R_d^{BF}$,
$\boldsymbol L_{\rm state}$, $L$, $r$th set of the selected relays,
{$S_{\rm total}$} and {$S$}}

\IF {$\boldsymbol L_{\rm state}(:,L)= \boldsymbol 0$}
\STATE $\eta_{\rm Link II}=0$

\COMMENT{The buffer is empty}
\ELSIF {$\boldsymbol L_{\rm state}(:,L)\neq \boldsymbol 0$}
{\color{black}\STATE $\boldsymbol{\Gamma}_{\boldsymbol \varOmega^{r,n}}^{(t)}= \boldsymbol I +(\boldsymbol H_{\boldsymbol
\varOmega^{r}}\boldsymbol R_{I}^{BF}\boldsymbol H_{\boldsymbol \varOmega^{r}}^{H})^{-1}(\boldsymbol H_{\boldsymbol
\varOmega^{r}}\boldsymbol R_{d}^{BF}\boldsymbol H_{\boldsymbol \varOmega^{r}}^{H})$
\STATE $\boldsymbol{\Gamma}_{\boldsymbol \varOmega^{r,e}}^{(t)}=\boldsymbol I+
\boldsymbol U_{r}^{H}{\boldsymbol R_{I}^{BF}}^{-1}\boldsymbol U_{r}\boldsymbol R_{d}^{BF}$
\STATE $ \boldsymbol\varGamma_{\rm II}({ \boldsymbol{\varOmega}^{r}})=
\sum
\bigg\{\log\big(\det{(\boldsymbol{\Gamma}_{\boldsymbol \varOmega^{r,n}}^{(t)})}\big) -\log\big(\det{(\boldsymbol{\Gamma}_{\boldsymbol \varOmega^{r,e}}^{(t)})}\big)\bigg\}
$}
\STATE $\eta_{\rm Link II}=\boldsymbol\varGamma_{\rm II}({ \boldsymbol{\varOmega}^{r}})$

\COMMENT {The buffer is not empty, the threshold for Link II $\eta_{\rm Link II}$ is calculated}
\ENDIF
\IF {$\boldsymbol L_{\rm state}(:,1)= \boldsymbol 0$}
\STATE $\boldsymbol \Psi_{\rm Jamming}=\binom{S_{\rm total}}{{\color{black}S}}$
\STATE $[{\varOmega_c}^{2} \quad {\varOmega_r}^{2}]={\rm size}(\boldsymbol \Psi_{\rm Jamming})$
\FOR {$m=1:{\varOmega_c}^{2}$}
{\color{black}
\STATE $\boldsymbol{\Gamma}_{\boldsymbol \varOmega^{m}}^{(t)} =\boldsymbol
I+ (\boldsymbol
H_{i}\boldsymbol R_{I}^{\boldsymbol \varOmega^{m}}\boldsymbol
H_{i}^{H})^{-1}(\boldsymbol
H_{i}\boldsymbol R_{d}^{\boldsymbol
\varOmega^{m}}\boldsymbol H_{i}^{H})$
\STATE
$\boldsymbol{\Gamma}_{\boldsymbol \varOmega^{e}}^{(t)}=\boldsymbol I+\boldsymbol U_{i}^{H}{\boldsymbol R_{I}^{\boldsymbol \varOmega^{m}}}^{-1}\boldsymbol U_{i}\boldsymbol R_{d}^{\boldsymbol \varOmega^{m}}$ \STATE $\boldsymbol\varGamma_{\rm I}(\boldsymbol{\varOmega}^{m})=
\sum
\bigg\{\log\big(\det{(\boldsymbol{\Gamma}_{\boldsymbol \varOmega^{m}}^{(t)})}\big) -\log\big(\det{(\boldsymbol{\Gamma}_{\boldsymbol \varOmega^{e}}^{(t)})}\big)\bigg\}
$}

\COMMENT {The buffer is not full, the threshold for Link I $\eta_{\rm Link I}$ is calculated}
\ENDFOR
\STATE $[\eta_{\rm Link I}, \boldsymbol{\varOmega}^{m,*}]=\rm{arg} \max_{\boldsymbol{\varOmega}^{m} \in
\boldsymbol{\Psi}_{\rm Jamming}} \boldsymbol\varGamma_{\rm I}(\boldsymbol{\varOmega}^{m})$
\ELSIF{$\boldsymbol L_{\rm state}(:,1)\neq \boldsymbol 0$}
\STATE $\eta_{\rm Link I}=0$

\COMMENT {The buffer is full}
\ENDIF
\IF{$\eta_{\rm Link II}>\eta_{\rm Link I}$}
\RETURN The set of the selected relays $\boldsymbol{\varOmega}^{r}$ and perform Link II.
\ELSIF{$\eta_{\rm Link II}<\eta_{\rm Link I}$}
\RETURN The set of the selected relays $\boldsymbol{\varOmega}^{m,*}$ and perform Link I.
\ENDIF

\end{algorithmic}
\end{algorithm}
}

\subsection{Proposed Greedy RJFS and BF-RJFS Algorithms}

In both RJFS and BF-RJFS algorithms, exhaustive searches are
implemented to select the relaying and jamming nodes. The
incorporation of a greedy strategy \cite{baplnc} in both RJFS and
BF-RJFS algorithms can significantly reduce the computational cost
of the proposed exhaustive search-based RJFS and BF-RJFS algorithms.
In an exhaustive search, all possible combinations are investigated
to achieve optimal relay selection. Unlike an exhaustive search, a
greedy search selects the relay with the best output at every
iteration and then repeats the process with the remaining relays.
The selection is completed when the desired number of relays are
chosen. The total number of relays considered in the search is given
by
\begin{equation}
\varOmega_c={\color{black}S_{\rm total}}+{\color{black}S_{\rm total}}-1+\cdots+{\color{black}S_{\rm total}}-S.
\label{eqn:greedycc}
\end{equation}
From (\ref{eqn:greedycc}), we can see that the number of relays
considered increases linearly with the total number of relays
$S_{\rm total}$ which contributes to the reduction of the
computational complexity.

In the following we describe the proposed greedy RJFS algorithm.
When the $K$ relays that forward the signals to the users are
determined, the relays used for signal reception are chosen based on
the E-SR criterion, as given by
\begin{equation}
m^*= \rm{arg} \max_{m \in
\boldsymbol{\Omega}}\big[\log\big(\det{(\boldsymbol{\Gamma}_{m}^{(t)})}\big)
-\log\big(\det{(\boldsymbol{\Gamma}_{e}^{(t)})}\big)\big],
\label{eqn:ORJSsinrct}
\end{equation}
where $m$ represents the selected relay and
$\boldsymbol{\Gamma}_{m}^{(t)}$ corresponds to the $m$th relay which
is calculated based on (\ref{eqn:SINRiwoc}) and given by
\begin{equation}
\boldsymbol{\Gamma}_{m}^{(t)} =\boldsymbol I + (\boldsymbol
H_{m}\boldsymbol R_{I}^{m}\boldsymbol H_{m}^{H})^{-1} (\boldsymbol
H_{m}\boldsymbol R_{d}^{m}\boldsymbol H_{m}^{H}), \label{eqn:gammam}
\end{equation}
whereas $\boldsymbol{\Gamma}_{e}^{(t)}$ is described by
\begin{equation}
\boldsymbol{\Gamma}_{e}^{(t)}=\boldsymbol I+\boldsymbol
U_{m}^{H}{\boldsymbol R_{I}^{m}}^{-1}\boldsymbol U_{m}\boldsymbol
R_{d}^{m}. \label{eqn:gammam2}
\end{equation}
{Instead of the exhaustive search of the selected set $\boldsymbol
\varOmega^{m}$, the $m$th relay is computed with the aim of finding
the relay that provides the highest secrecy rate based on $m$. In
(\ref{eqn:gammam}), $\boldsymbol R_{d}^{m}$ and $\boldsymbol
R_{I}^{m}$ are obtained in the same way as in (\ref{eqn:Omega4}) and
(\ref{eqn:Omega8}).} The main steps are described in
Algorithm~\ref{alg:grjfs}.

\begin{algorithm}
\caption{Greedy-RJFS Algorithm}
\label{alg:grjfs}
\begin{algorithmic}[1]
{\REQUIRE $\boldsymbol H_m$, precoding matrix $\boldsymbol U_m$,
$\boldsymbol R_I^m$, $\boldsymbol R_d^m$, $S$ $\boldsymbol \Omega$}

\FOR {$t=1:S$} \STATE $\varOmega=\rm{length}(\boldsymbol \Omega)$
\FOR {$m=1:\varOmega$} {\color{black} \STATE
$\boldsymbol{\Gamma}_{m}^{(t)} =\boldsymbol I + (\boldsymbol
H_{m}\boldsymbol R_{I}^{m}\boldsymbol H_{m}^{H})^{-1}(\boldsymbol
H_{m}\boldsymbol R_{d}^{m}\boldsymbol H_{m}^{H})$ \STATE
$\boldsymbol{\Gamma}_{e}^{(t)}=\boldsymbol I+\boldsymbol
U_{m}^{H}{\boldsymbol R_{I}^{m}}^{-1}\boldsymbol U_{m}\boldsymbol
R_{d}^{m} $ \STATE $\boldsymbol\varGamma(m)=
[\log\big(\det{(\boldsymbol{\Gamma}_{m}^{(t)})}\big)
-\log\big(\det{(\boldsymbol{\Gamma}_{e}^{(t)})}\big)] $}

\COMMENT{Calculate the threshold for all relays}
\ENDFOR
\STATE $m^*=\rm{arg} \max_{m }(\boldsymbol\varGamma(m))$

\COMMENT{Find the relay which gives the highest value and choose this relay as one of the selected relay node}
\STATE $\boldsymbol {\varOmega}^{\rm Greedy,*}=m^*$
\STATE $\boldsymbol \Omega=\boldsymbol \Omega/m^*$
\STATE $\boldsymbol \varGamma=\boldsymbol \varGamma/m^*$

\COMMENT{Remove the selected relay node from all relay set.}
\ENDFOR

\COMMENT{Repeat the steps again until $S$ relays are found}
\rev{\RETURN The set of the selected relays
$\boldsymbol{\varOmega}^{\rm Greedy,*}$}
\end{algorithmic}
\end{algorithm}

Similarly to the BF-RJFS algorithm, the Greedy-BF-RJFS algorithm
substitutes the exhaustive search of all combinations with a greedy
search of individual relays. The main difference lies in the jamming
relay selection. For a particular user $r$, each relay performs a
threshold calculation and the relay $k$ with the highest threshold
is selected until $S$ relays are selected to forward the signal to
all users. The details of the Greedy-BF-RJFS algorithm are given in
Algorithm~\ref{alg:gbr}.

\begin{algorithm}
\caption{Greedy-BF-RJFS Algorithm}
\label{alg:gbr}
\begin{algorithmic}[1]
{\REQUIRE $\boldsymbol H_m$, $\boldsymbol R_I^m$, $\boldsymbol
R_d^m$, $\boldsymbol H_{\boldsymbol \Omega^r}$ and precoding matrix
$\boldsymbol U_r$, $\boldsymbol R_I^{BF}$, $\boldsymbol R_d^{BF}$,
$\boldsymbol L_{\rm state}$, $L$, $\boldsymbol \varOmega^{r}$, $S$
and $\boldsymbol \Omega$}

\IF {$\boldsymbol L_{\rm state}(:,L)= \boldsymbol 0$}

\STATE $\eta_{\rm Link II}=0$

\COMMENT{The buffer is empty}
\ELSIF {$\boldsymbol L_{\rm state}(:,L)\neq \boldsymbol 0$}
{\color{black}\STATE $\boldsymbol{\Gamma}_{\boldsymbol \varOmega^{r,n}}^{(t)}= \boldsymbol I +(\boldsymbol H_{\boldsymbol
\varOmega^{r}}\boldsymbol R_{I}^{BF}\boldsymbol H_{\boldsymbol \varOmega^{r}}^{H})^{-1}(\boldsymbol H_{\boldsymbol
\varOmega^{r}}\boldsymbol R_{d}^{BF}\boldsymbol H_{\boldsymbol \varOmega^{r}}^{H})$
\STATE $\boldsymbol{\Gamma}_{\boldsymbol \varOmega^{r,e}}^{(t)}=\boldsymbol I+
\boldsymbol U_{r}^{H}{\boldsymbol R_{I}^{BF}}^{-1}\boldsymbol U_{r}\boldsymbol R_{d}^{BF}$
\STATE $ \boldsymbol\varGamma_{\rm II}({ \boldsymbol{\varOmega}^{r}})=
\sum
\bigg\{\log\big(\det{(\boldsymbol{\Gamma}_{\boldsymbol \varOmega^{r,n}}^{(t)})}\big) -\log\big(\det{(\boldsymbol{\Gamma}_{\boldsymbol \varOmega^{r,e}}^{(t)})}\big)\bigg\}
$}
\STATE $\eta_{\rm Link II}=\boldsymbol\varGamma_{\rm II}({ \boldsymbol{\varOmega}^{r}})$

\COMMENT {The buffer is not empty, the threshold for Link II
$\eta_{\rm Link II}$ is calculated} \ENDIF \IF {$\boldsymbol L_{\rm
state}(:,1)= \boldsymbol 0$} \FOR {$t=1:S$} \FOR {$m=1:\varOmega$}
{\color{black} \STATE $\boldsymbol{\Gamma}_{m}^{(t)} =\boldsymbol I
+ (\boldsymbol H_{m}\boldsymbol R_{I}^{m}\boldsymbol
H_{m}^{H})^{-1}(\boldsymbol H_{m}\boldsymbol R_{d}^{m}\boldsymbol
H_{m}^{H})$ \STATE $\boldsymbol{\Gamma}_{e}^{(t)}=\boldsymbol
I+\boldsymbol U_{m}^{H}{\boldsymbol R_{I}^{m}}^{-1}\boldsymbol
U_{m}\boldsymbol R_{d}^{m} $ \STATE $\boldsymbol\varGamma_{\rm
I}(m)= [\log\big(\det{(\boldsymbol{\Gamma}_{m}^{(t)})}\big)
-\log\big(\det{(\boldsymbol{\Gamma}_{e}^{(t)})}\big)] $} \ENDFOR
\STATE $[\eta_{\rm Link I},m^*]=\rm{arg} \max_{m
}(\boldsymbol\varGamma_{\rm I}(m))$ \STATE $\boldsymbol
{\varOmega}^{\rm Greedy,*}=m^*$ \STATE $\boldsymbol
\Omega=\boldsymbol \Omega/m^*$ \STATE $\boldsymbol \varGamma_{\rm
I}=\boldsymbol \varGamma_{\rm I}/m^*$ \ENDFOR

\COMMENT{The same steps as Greedy-RJFS algorithm to select $S$
relays out of all relay set $\boldsymbol \Omega$} \ELSIF
{$\boldsymbol L_{\rm state}(:,1)\neq \boldsymbol 0$} \STATE
$\eta_{\rm Link I}=0$

\COMMENT {The buffer is full}
\ENDIF
\IF{$\eta_{\rm Link II}>\eta_{\rm Link I}$}
\RETURN The set of the selected relays $\boldsymbol{\varOmega}^{r}$ and perform Link II.
\ELSIF{$\eta_{\rm Link II}<\eta_{\rm Link I}$}
\RETURN The set of the selected relays $\boldsymbol{\varOmega}^{\rm Greedy,*}$ and perform Link I.
\ENDIF

\end{algorithmic}
\end{algorithm}

\section{Secrecy Analysis}

In this section, we analyze the secrecy performance of standard
single-antenna and MIMO relay systems as well as the proposed
buffer-aided MIMO relay system with relaying and jamming function
selection. We derive secrecy rate expressions for scenarios where
CSI is available to the eavesdroppers. The expressions derived serve
as benchmarks for the proposed RJFS and BF-RJFS algorithms. The
overall secrecy capacity of a single-antenna relay system
\cite{Gaojie} is given by
\begin{myDef}
For a selected relay $k$ and channels from source to relay $k$,
relay $k$ to destination, source to eavesdropper, relay $k$ to
eavesdropper expressed as
$h_{s{r_{k}}},h_{{r_{k}}d},h_{se},h_{{r_{k}}e}$ respectively, the
capacity is given by
\begin{equation}
C_{k}=\max \Big\{\frac{1}{2}\log_{2}{\frac{\min\{1 +
P{\|h_{s{r_{k}}}\|}^{2},1+P{\|h_{{r_{k}}d}\|}^{2}\}}{1+P{\|h_{se}\|}^{2}+P{\|h_{{r_{k}}e}\|}^{2}}}
\Big\} \label{equ:CDk}
\end{equation}
\end{myDef}
Equation (\ref{equ:CDk}) can be rewritten as (\ref{equ:CDk2}) in
which the first part corresponds to the secrecy capacity to the user
and the second part to the secrecy capacity of the eavesdropper:
\begin{align}
C_{k} & =  \max \bigg[\frac{1}{2}\log_{2}\big({{\min\{1+P{\|h_{s{r_{k}}}\|}^{2},1+P{\|h_{{r_{k}}d}\|}^{2}\}}}\big) \nonumber \\
 & \qquad
 -\frac{1}{2}\log_{2}\big({1+P{\|h_{se}\|}^{2}+P{\|h_{{r_{k}}e}\|}^{2}}\big)\bigg].
\label{equ:CDk2}
\end{align}
In half-duplex MIMO relay systems, based on (\ref{eqn:RR2}) and
(\ref{eqn:RE2}), the secrecy capacity from the source to the relay
and to the eavesdropper can be respectively expressed by
\begin{equation}
C_{i}=\max \bigg[\frac{1}{2}\log_{2}\big({{\det{(\boldsymbol I + \boldsymbol
H_{i}^{(t)}\boldsymbol U \boldsymbol U^H {\boldsymbol H_{i}^{(t)}}^{H})}}}\big)\bigg] \label{equ:CDk4}
\end{equation}
\begin{equation}
C_{e}=\max \bigg[\frac{1}{2}\log_{2}\big(\det({\boldsymbol I +
\boldsymbol{\Gamma}_{e}^{(t)}})\big)\bigg] \label{equ:CDk5}
\end{equation}
The secrecy capacity from relay to destination is given by
\begin{equation}
C_{r}=\max \bigg[\frac{1}{2}\log_{2}\big({{\det(\boldsymbol
I+\boldsymbol{\Gamma}_{r}^{(t)})}}\big)\bigg]. \label{equ:CDk6}
\end{equation}
With equations (\ref{equ:CDk4}), (\ref{equ:CDk5}) and
(\ref{equ:CDk6}) based on the overall secrecy capacity of
single-antenna relay systems, we can express the overall secrecy
capacity of MIMO relay systems:
\begin{align}
C_{k}^{\rm {MIMO}} & =  \max \bigg[\frac{1}{2}\log_{2}\big({{\min \{M_{i},M_{r}\}}}\big) \nonumber \\
 & \qquad -\frac{1}{2}\log_{2}\big(\det({\boldsymbol
 I+\boldsymbol{\Gamma}_{e}^{(t)}})\big)\bigg],
\label{equ:CDk3}
\end{align}
where $M_{i}=\det{(\boldsymbol I+\boldsymbol H_{i}^{(t)}\boldsymbol
U \boldsymbol U^H {\boldsymbol H_{i}^{(t)}}^{H})}$ and
$M_{r}=\det(\boldsymbol I+\boldsymbol{\Gamma}_{r}^{(t)})$. Note that
the factor $\frac{1}{2}$ is due to half-duplex systems.
\begin{myPro}
With buffers of size $L$ implemented in the relay nodes, the
secrecy-rate performance can be improved. The secrecy rate
difference varies between $0$ to $\varDelta_{\rm{R-BF}}$.
\end{myPro}
\begin{proof}
In half-duplex MIMO relay systems with multiple relays, relay
selection can be performed prior to transmission. If we use
$\boldsymbol{\Psi}$ to represent a set of relay nodes based on
(\ref{equ:CDk3}) then with relay selection the secrecy rate is
expressed by
\begin{equation}
\max_{i \in \boldsymbol{\Psi}}{{\min \{\det{(\boldsymbol I
+\boldsymbol H_{i}^{(t)} \boldsymbol U \boldsymbol U^H{\boldsymbol
H_{i}^{(t)}}^{H})},\det(\boldsymbol
I+\boldsymbol{\Gamma}_{r}^{(t)})\}}} \label{equ:CDk7}
\end{equation}
Under the condition that $\det{(\boldsymbol I+\boldsymbol
H_{i}^{(t)}\boldsymbol U \boldsymbol U^H {\boldsymbol
H_{i}^{(t)}}^{H})} < \det(\boldsymbol
I+\boldsymbol{\Gamma}_{r}^{(t)})$, relay selection can be simplified
and given by
\begin{equation}
\max_{i_{R} \in \boldsymbol{\Psi}}{{ \{\det{(\boldsymbol I + \boldsymbol
H_{i}^{(t)} \boldsymbol U \boldsymbol U^H {\boldsymbol H_{i}^{(t)}}^{H})}\}}}, \label{equ:CDk8}
\end{equation}
where $i_{R}$ represents the selected relay. In this scenario, the
secrecy rate is described by
\begin{align}
C_{\rm{Relay}}^{(1)} & = \frac{1}{2}\log_{2}\big({{ \{\det{(\boldsymbol I+\boldsymbol H_{i_{R}}^{(t)} \boldsymbol U \boldsymbol U^H{\boldsymbol H_{i_{R}}^{(t)}}^{H})}\}}}\big) \nonumber \\
 & \qquad -\frac{1}{2}\log_{2}\big(\det({\boldsymbol
 I+\boldsymbol{\Gamma}_{e}^{(t)}})\big).
\label{equ:CDk9}
\end{align}
Under the condition that $\det{(\boldsymbol I+\boldsymbol
H_{i}^{(t)}\boldsymbol U \boldsymbol U^H {\boldsymbol
H_{i}^{(t)}}^{H})} > \det(\boldsymbol
I+\boldsymbol{\Gamma}_{r,i}^{(t)})$, the secrecy rate can be
computed in the same way as above and the result is given by
\begin{align}
C_{\rm{Relay}}^{(2)} & = \frac{1}{2}\log_{2}\big({{ \{\det(\boldsymbol I+\boldsymbol{\Gamma}_{r}^{(t)})\}}}\big) \nonumber \\
 & \qquad -\frac{1}{2}\log_{2}\big(\det({\boldsymbol
 I+\boldsymbol{\Gamma}_{e}^{(t)}})\big).
\label{equ:CDk10}
\end{align}
When each relay node is equipped with an infinite buffer, the signals can be
stored in the buffers which means the signals can wait at the relay nodes until
the condition $\det{(\boldsymbol I+\boldsymbol H_{i}^{(t)}\boldsymbol U \boldsymbol U^H {\boldsymbol
H_{i}^{(t)}}^{H})} < \det(\boldsymbol I+\boldsymbol{\Gamma}_{r}^{(t)})$ is
satisfied. If we use $\det(\boldsymbol I+\boldsymbol{\Gamma}_{r}^{(pt)})$ to
represent the condition that is experienced in the previous time slot which
follows $\det{(\boldsymbol I+\boldsymbol H_{i}^{(pt)} \boldsymbol U^{(pt)}  {\boldsymbol U^{(pt)} }^H{\boldsymbol
H_{i}^{(pt)}}^{H})}
> \det(\boldsymbol I+\boldsymbol{\Gamma}_{r}^{(pt)})$, then the expression of the
secrecy rate with infinite buffers is described by
\begin{align}
\varDelta_{\rm{R-BF}} & = \frac{1}{2}\log_{2}\big({{ \{\det{(\boldsymbol I+\boldsymbol H_{i_{R}}^{(t)}\boldsymbol U \boldsymbol U^H {\boldsymbol H_{i_{R}}^{(t)}}^{H})}\}}}\big) \nonumber \\
 & \qquad -\frac{1}{2}\log_{2}\big({{ \{\det(\boldsymbol
 I+\boldsymbol{\Gamma}_{r}^{(pt)})\}}}\big).
\label{equ:CDk11}
\end{align}
Specifically, with a buffer with size $L$ the condition
$\det{(\boldsymbol I+\boldsymbol H_{i}^{(t)} \boldsymbol U
\boldsymbol U^H{\boldsymbol H_{i}^{(t)}}^{H})} > \det(\boldsymbol
I+\boldsymbol{\Gamma}_{r}^{(t)})$ will not hold and the difference
of the secrecy rates will be between $0$ and
$\varDelta_{\rm{R-BF}}$.
\end{proof}
In the scenarios considered, to avoid the interference in the
transmission to or from the relays, a half-duplex scheme is
employed. To limit the number of time slots, an opportunistic scheme
can be applied to MIMO relay systems.
\begin{myPro}
An opportunistic scheme can improve the secrecy rate as compared
with standard half-duplex MIMO relay systems.
\end{myPro}
\begin{proof}
According to \cite{Nomikos2}, in the opportunistic scheme we have
concurrent transmissions with all relays. This will result in IRI
and as a result its effect on the relay that receives the source
signal must be considered during the opportunistic scheme. In
\cite{Nomikos2}, it has been pointed out that IC can be performed at
the relay node. To simplify the proof, we first assume IC is
performed and the secrecy rate is expressed by
\begin{equation}
C_{\rm{Opportunistic-Relay}}^{(1)}=2 \times C_{\rm{Relay}}^{(1)}
\label{equ:CDk12},
\end{equation}
\begin{equation}
C_{\rm{Opportunistic-Relay}}^{(2)}=2 \times C_{\rm{Relay}}^{(2)},
\label{equ:CDk13}
\end{equation}
and
\begin{equation}
\varDelta_{\rm{Opportunistic-Relay-buffer}}=2 \times
\varDelta_{\rm{Relay-buffer}}, \label{equ:CDk14}
\end{equation}
which shows that the secrecy rate of the opportunistic scheme
doubles. If IC cannot be performed, based on (\ref{equ:CDk9}), the
secrecy rate is expressed by
\begin{align}
\hspace{-1em} C_{\rm{Opp-Relay}}^{(1)} & =\log_{2}\big({{
\{\det{(\boldsymbol I+({\boldsymbol{I}+\boldsymbol
\varDelta_{i_{R}}'})^{-1}\boldsymbol H_{i_{R}}^{(t)}\boldsymbol U
\boldsymbol U^H {\boldsymbol
H_{i_{R}}^{(t)}}^{H})}\}}}\big) \nonumber \\
 & \qquad -\log_{2}\big(\det({\boldsymbol
 I+\boldsymbol{\Gamma}_{e}^{(t)}})\big),
\label{equ:CDk15}
\end{align}
where
\begin{equation}
\boldsymbol \varDelta_{i_{R}}'=\sum_{k=1}^{K}\boldsymbol{H}_{ki_{R}}
\boldsymbol{H}_{i_{R}}^{(pt)}\boldsymbol U^{(pt)} {\boldsymbol U^{(pt)}}^H{\boldsymbol{H}_{i_{R}}^{(pt)}}^{H}\boldsymbol{H}_{ki_{R}}^{H},
\label{eqn:gammaR}
\end{equation}
which represents IRI. Then, the secrecy rate difference between a
standard relay system and an opportunistic buffer-aided relay system
is obtained by
\begin{align}
\hspace{-1em}\varDelta_{\rm{Opp-R-BF}} & =\log_{2}\big({{
\{\det{(\boldsymbol I+({\boldsymbol{I}+\boldsymbol
\varDelta_{i_{R}}'})^{-1}\boldsymbol
H_{i_{R}}^{(t)} \boldsymbol U \boldsymbol U^H {\boldsymbol H_{i_{R}}^{(t)}}^{H})}\}}}\big) \nonumber \\
 & \qquad -\log_{2}\big({{
 \{\det(\boldsymbol I+\boldsymbol{\Gamma}_{r}^{(pt)})\}}}\big).
\label{equ:CDk16}
\end{align}
\end{proof}

\subsection{Relaying and Jamming Function Selection}

\begin{myThe}
When $\rm{SNR}\rightarrow\infty$, the secrecy rate
$C_{\rm{Opportunstic-Relay-buffer}}\rightarrow\infty$ and the
secrecy rate with IRI cancellation outperforms that without IRI
cancellation.
\end{myThe}

\begin{proof}
In all aforementioned systems, we have not taken any jamming signal
into consideration. In the presence of systems with multiple relay
nodes, some relay nodes can perform the jamming function by
transmitting jamming signals to the eavesdroppers. More
specifically, IRI cancellation is considered. In the RJFS algorithm,
the selected relay at the current time interval is the jammer as
well as the relay responsible for forwarding the data in the next
time interval. The aim of the RJFS algorithm is to choose the relay
that provides the highest secrecy rate performance.

According to Algorithm 1, the relay selection criterion is given by
\begin{equation}
{\mathcal R}_{i_{R}}=\rm{arg} \max_{i_{R} \in \boldsymbol{\Psi}}
\det \left(
({\boldsymbol{I}+\boldsymbol{\Gamma}_{e}^{(t)}})^{-1}({\boldsymbol{I}+\boldsymbol{\Gamma}_{i_{R}}^{(t)}})\right)
\label{equ:CDk17}
\end{equation}
where $i_{R}$ represents the selected relay. Based on
(\ref{equ:CDk17}), the secrecy rate with the selected relay can be
expressed as:
\begin{align}
C_{\rm{RJFS-IRI}} & =\log_{2}\big({{ \{\det{(\boldsymbol I+\boldsymbol{\Gamma}_{i_{R}}^{(t)})}\}}}\big) \nonumber \\
 & \qquad -\log_{2}\big({{ \{\det(\boldsymbol
 I+\boldsymbol{\Gamma}_{e}^{(t)})\}}}\big).
\label{equ:CDk18}
\end{align}
When IC is performed at the relay nodes, (\ref{equ:CDk18}) can be
simplified to
\begin{align}
C_{\rm{RJFS-IC}} & =\log_{2}\big({{ \{\det{(\boldsymbol I+\boldsymbol H_{i_{R}}^{(t)} \boldsymbol U \boldsymbol U^H{\boldsymbol H_{i_{R}}^{(t)}}^{H})}\}}}\big) \nonumber \\
 & \qquad -\log_{2}\big({{ \{\det(\boldsymbol
 I+\boldsymbol{\Gamma}_{e}^{(t)})\}}}\big).
\label{equ:CDk19}
\end{align}
Equation (\ref{equ:CDk19}) was obtained in our previous study
\cite{Xiaotao1} when $\rm{SNR}\rightarrow\infty$. Comparing
(\ref{equ:CDk18}) with (\ref{equ:CDk19}), we can have
$C_{\rm{RJFS-IC}} > C_{\rm{RJFS-IRI}}$, as indicated in Fig.
\ref{fig:IRI1}.
\end{proof}

\subsection{Buffer-aided Relay and Jammer Function Selection}

\begin{myThe}
According to Proposition 1, the secrecy-rate performance can be
improved with buffers. This can also be applied to the RJFS
algorithm. In the IC scenario, when more power is allocated to the
transmitter the secrecy rate will suffer from a dramatic decrease.
\end{myThe}
\begin{proof}
In the buffer-aided RJFS algorithm, the relay selection and jamming
selection can be implemented simultaneously with the following
selection criterion:
\begin{equation}
{\mathcal R}_{i_{R}}=\rm{arg} \max_{i_{R} \in \boldsymbol{\Psi}}
\det \left(
({\boldsymbol{I}+\boldsymbol{\Gamma}_{e}^{(t)}})^{-1}({\boldsymbol{I}+\boldsymbol{\Gamma}_{i_{R}}^{(t)}})\right)
\end{equation}
and
\begin{equation}
{\mathcal R}_{n}= \rm{arg} \max_{n \in \boldsymbol{\Psi}} \det
\left(
({\boldsymbol{I}+\boldsymbol{\Gamma}_{e}^{(t)}})^{-1}({\boldsymbol{I}+\boldsymbol{\Gamma}_{n}^{(t)}})\right),
\end{equation}
where in both transmissions we can achieve high secrecy rate
performance with separate selection from the source to the relays
and from the relays to the destination. Considering power
allocation, with the parameter $\eta$ indicating the power allocated
to the transmitter, we assume the power allocated to the transmitter
is $\eta P$ and the power allocated to the relays is $(2-\eta)P$.
When $\eta \rightarrow 0 $, less power will be allocated to the
transmitter according to:
\begin{align}
C_{\rm{BF-RJFS-IRI}}^{(1)} & = \log_{2}\big({{ \{\det{(\boldsymbol I+\boldsymbol{\Gamma}_{i_{R}}^{(t)})}\}}}\big) \nonumber \\
 & \qquad -\log_{2}\big({{ \{\det(\boldsymbol I+\boldsymbol{\Gamma}_{e}^{(t)})\}}}\big)
\label{equ:CDk20}
\end{align}
and
\begin{align}
C_{\rm{BF-RJFS-IRI}}^{(2)} & = \log_{2}\big({{ \{\det{(\boldsymbol I+\boldsymbol{\Gamma}_{n}^{(t)})}\}}}\big) \nonumber \\
 & \qquad -\log_{2}\big({{ \{\det(\boldsymbol
 I+\boldsymbol{\Gamma}_{e}^{(t)})\}}}\big),
\label{equ:CDk21}
\end{align}
where the secrecy rate $C_{\rm{BF-RJFS-IRI}}^{(1)}$ will have an
increase, while $C_{\rm{BF-RJFS-IRI}}^{(2)}$ will have a decrease.
As a result, the overall secrecy rate will decrease. More
specifically, with less power allocated in the relay or the jammer,
IC that acts as a jamming signal to the eavesdropper will have less
effect on the contribution to the secrecy rate. Then, the overall
secrecy rate will have a dramatic decrease.

If IC is performed, according to (76) and (77), we have
\begin{align}
C_{\rm{BF-RJFS-IC}}^{(1)} & = \log_{2}\big({{ \{\det{(\boldsymbol I+\boldsymbol H_{i_{R}}^{(t)} \boldsymbol U \boldsymbol U^H {\boldsymbol H_{i_{R}}^{(t)}}^{H})}\}}}\big) \nonumber \\
 & \qquad -\log_{2}\big({{ \{\det(\boldsymbol I+\boldsymbol{\Gamma}_{e}^{(t)})\}}}\big)
\label{equ:CDk22}
\end{align}
and
\begin{align}
C_{\rm{BF-RJFS-IC}}^{(2)} & = \log_{2}\big({{ \{\det{(\boldsymbol I+\boldsymbol{\Gamma}_{n}^{(t)})}\}}}\big) \nonumber \\
 & \qquad -\log_{2}\big({{ \{\det(\boldsymbol
 I+\boldsymbol{\Gamma}_{e}^{(t)})\}}}\big),
\label{equ:CDk23}
\end{align}
where more power is allocated to the transmitter and the secrecy
rates $C_{\rm{BF-RJFS-IC}}^{(1)}$ and $C_{\rm{BF-RJFS-IC}}^{(2)}$
are less affected than those in the scenario with IC. The results in
Figs. \ref{fig:pa1} and \ref{fig:pa2} indicate the change with
different power allocation.
\end{proof}

\subsection{Greedy Algorithm}

\begin{myThe}
With high $\rm{SNRs}$, the proposed greedy BF-RJFS algorithm can
achieve comparable secrecy rate performance with a dramatic
reduction in the computational cost.
\end{myThe}

\begin{proof}
According to (\ref{eqn:Omega}), the total number of visited sets for
an exhaustive search can be expressed as:
\begin{equation}
\varOmega_{\rm exhaustive}=\dfrac{{\color{black}S_{\rm
total}}!}{({\color{black}S_{\rm
total}}-{\color{black}S})!{\color{black}S}!}. \label{eqn:exh}
\end{equation}
In the proposed greedy algorithms, the search is implemented in the
remaining relay nodes so that the total number of visited sets in
the greedy search is given by
\begin{equation}
\begin{split}
\varOmega_{\rm greedy}& ={\color{black}S_{\rm
total}}+{\color{black}S_{\rm total}}-1 +\cdots+{\color{black}S_{\rm
total}}-{\color{black}S}+1 \\ & ={\color{black}S_{\rm
total}}{\color{black}S}-\dfrac{{\color{black}S}({\color{black}S}-1)}{2}
\label{eqn:gre}
\end{split}
\end{equation}
Based on (\ref{eqn:exh}) and (\ref{eqn:gre}), for a number of
selected relay nodes {\color{black}$S$}, when the total number of
relay nodes $S_{\rm out}$ increases, the total number of visited
sets for the exhaustive search is much higher than those for the
greedy search, that is $\varOmega_{\rm exhaustive}>> \varOmega_{\rm
greedy}$.
\end{proof}

\section{Simulation Results}

In this section, we assess the secrecy-rate performance of the
proposed E-SR relay selection criterion and the RJFS and BF-RJFS
algorithms against existing techniques via simulations for the
downlink of a multiuser buffer-aided relay systems. In particular,
the proposed E-SR relay selection criterion is compared against the
impractical SR method that uses the CSI of the eavesdroppers, the
SINR-based techniques and the max-ratio approach. Moreover, the
proposed RJFS and BF-RJFS algorithms that employ IC are evaluated
against an approach without IC. We consider both single-antenna and
MIMO settings. In a single-antenna scenario, the transmitter is
equipped with $3$ antennas to broadcast the signal to $3$ legitimate
users through multiple single-antenna relays in the presence of $3$
eavesdroppers equipped with a single antenna. In the MIMO scenario,
the transmitter is equipped with $6$ antennas and each user,
eavesdropper and relay has $2$ antennas. The buffer can store up to
$J = 4$ packets. In both scenarios, a zero-forcing precoding
technique is employed at the transmitter, $\eta_{L_{\rm max}}=5$ and
we assume that the CSI for each user is available.

In the first example, we consider a scenario with uncorrelated
channels, whereas in the second example we include correlated
channels whose channel matrix is expressed by
\begin{equation}
\boldsymbol H_c=\boldsymbol {R_r}^{\frac{1}{2}}\boldsymbol H\boldsymbol {R_t}^{\frac{1}{2}}
\label{eqn:Hc}
\end{equation}
where $\boldsymbol {R_r}$ and $\boldsymbol {R_t}$ are receive and
transmit covariance matrices with ${\rm{Tr}}(\boldsymbol {R_r}) =
Nr$ and ${\rm{Tr}}(\boldsymbol {R_t}) = Nt$. Both $\boldsymbol
{R_r}$ and $\boldsymbol {R_t}$ are positive semi-definite Hermitian
matrices. For the case of an urban wireless environment, the user
is always surrounded by rich scattering objects and the channel is
most likely independent Rayleigh fading at the receive side. Hence,
we assume $\boldsymbol {R_r} = \boldsymbol {I}_{N_r}$, and we have
\begin{equation}
\boldsymbol H_c=\boldsymbol H\boldsymbol {R_t}^{\frac{1}{2}}
\label{eqn:Hc1}
\end{equation}
To study the effect of antenna correlations, random realizations of
correlated channels are generated based on the exponential
correlation model such that the elements of $\boldsymbol {R_t}$ are
given by
\begin{equation}
\boldsymbol {R_t}(i,j)=\begin{cases}
    r^{j-i}       & \quad \text{if } i\leq j\\
    r_{j,i}^{*}  & \quad \text{if } i>j\\
  \end{cases}, |r|\leq 1
\label{eqn:Hcr}
\end{equation}
where $r$ is the correlation coefficient between any two neighboring
antennas.

\begin{figure}[ht]
\centering
\includegraphics[width=0.85\linewidth]{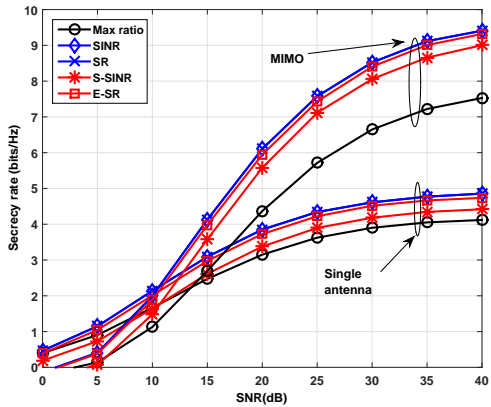}
\caption{Secrecy-rate performance of relay selection criteria in
uncorrelated channels.} \label{fig:thr}
\end{figure}

\begin{figure}[ht]
\centering
\includegraphics[width=0.85\linewidth]{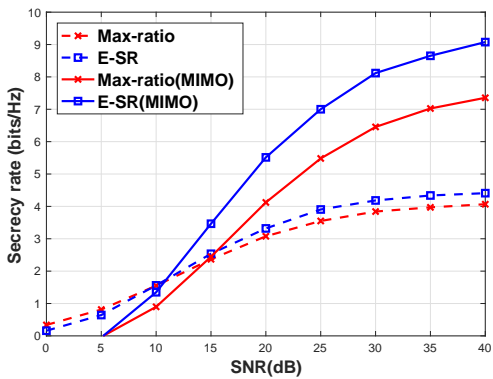}
\caption{Secrecy-rate performance of relay selection criteria in
correlated channels.} \label{fig:cor}
\end{figure}

In Figs. \ref{fig:thr} and \ref{fig:cor} we compare the secrecy rate
performance in uncorrelated and correlated channels. The results
indicate that the proposed E-SR relay selection criterion can
improve the secrecy rate in both scenarios. Among the investigated
relay selection criteria, E-SR is close to the SR-based scheme that
employs CSI to the eavesdroppers and outperforms SINR-based
techniques, which are often adopted in the literature
\cite{Dong,Shafie1} and require the CSI of the eavesdroppers.

\begin{figure}[ht]
\centering
\includegraphics[width=0.85\linewidth]{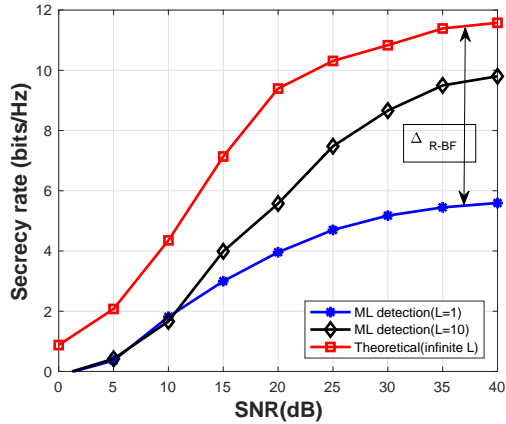}
\caption{Secrecy-rate performance versus buffer size in uncorrelated
channels.} \label{fig:theo}
\end{figure}

In Fig. \ref{fig:theo}, the secrecy rate performance with infinite
buffer size is compared with buffer size $L=1$ and $L=10$. The
theoretical curves are obtained with the expression obtained in
Section V for the secrecy rate difference $\Delta_{\rm R-BF}$.
According to the results, when the buffer size is increased, the
secrecy rate will improve and get close to the theoretical curves.

\begin{figure}[ht]
\centering
\includegraphics[width=0.85\linewidth]{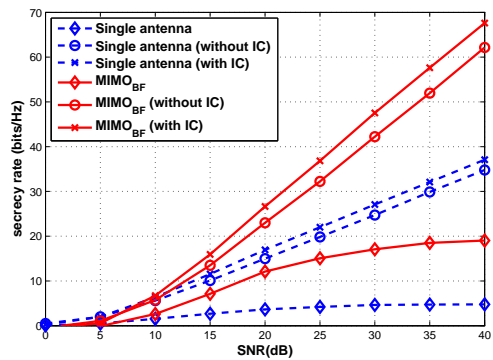}
\caption{Secrecy rate performance in correlated channels.}
\label{fig:IRI1}
\end{figure}

In Fig. \ref{fig:IRI1}, in a single-antenna scenario, the
secrecy-rate performance with the proposed IC scheme and RJFS
algorithm is better than that with the conventional algorithm
without IC. With IC, the secrecy-rate performance is better than the
one without IC, as expected. Compared with the single-antenna
scenario, the multiuser MIMO system contributes to the improvement
in the secrecy rate as verified in Fig. \ref{fig:IRI1}.

\begin{figure}[ht]
\centering
\includegraphics[width=0.85\linewidth]{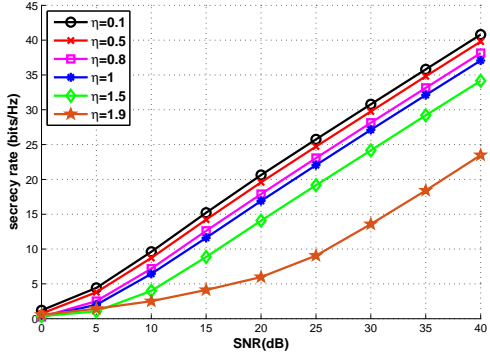}
\caption{Secrecy rate performance with power allocation and IC.}
\label{fig:pa1}
\end{figure}

\begin{figure}[ht]
\centering
\includegraphics[width=0.85\linewidth]{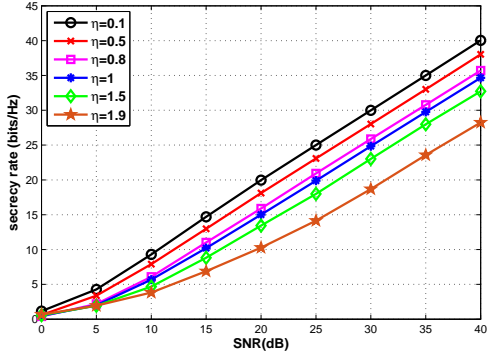}
\caption{Secrecy rate performance with power allocation and without
IC.} \label{fig:pa2}
\end{figure}

In Fig. \ref{fig:pa1} and Fig. \ref{fig:pa2}, a power allocation
technique is considered and the parameter $\eta$ indicates the power
allocated to the transmitter. If we assume in the equal power
scenario that the power allocated to the transmitter as well as the
relays are both $P$, then the power allocated to the transmitter is
$\eta P$ and the power allocated to the relays is $(2-\eta)P$. In
Fig. \ref{fig:pa1} and Fig. \ref{fig:pa2} we can notice that with
more power allocated to the transmitter the secrecy rate performance
will become worse. Comparing Fig. \ref{fig:pa1} and Fig.
\ref{fig:pa2}, when $\eta<1.5$ the secrecy rate performance in the
scenario with IC is better than that without IC. When $\eta>1.5$ the
secrecy rate of the system without IC is better than that of the
system with IC.

\begin{figure}[ht]
\centering
\includegraphics[width=0.85\linewidth]{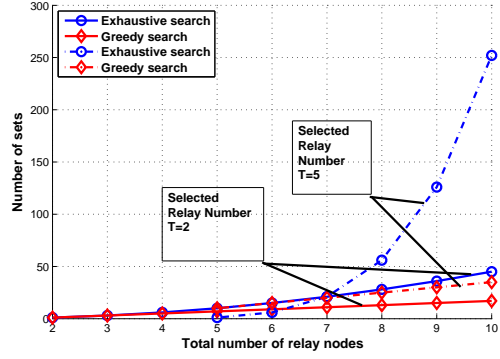}
\caption{Number of visited sets for the exhaustive and greedy
searches.} \label{fig:complexity}
\end{figure}

\begin{figure}[ht]
\centering
\includegraphics[width=0.85\linewidth]{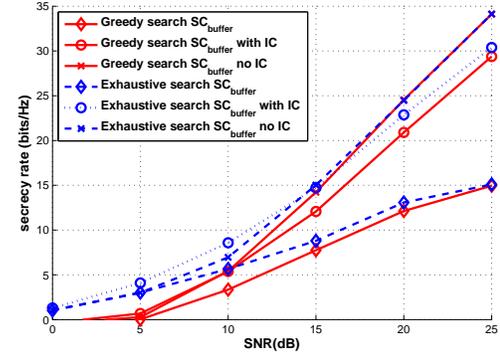}
\caption{Secrecy rate performance with an exhaustive search and the
proposed greedy algorithms.} \label{fig:greedy}
\end{figure}

In Fig. \ref{fig:complexity} with a fixed number of relays, the
computational complexity of the exhaustive and the greedy searches
with the RJFS and BF-RJFS algorithms is examined. The results show
that the greedy algorithms are substantially simpler than those with
the exhaustive search and are suitable for scenarios with a higher
number of relays. In Fig. \ref{fig:greedy} a comparison between the
exhaustive search and the greedy algorithms is carried out. The
results show that the greedy algorithms approach the same secrecy
rate with a much lower complexity than that of the exhaustive
search-based techniques.

\section{Conclusion}

In this work, we have proposed the E-SR approach that allows the
maximization of the secrecy rate in buffer-aided relay systems
without the need for the CSI of the eavesdroppers. We have also
presented algorithms to select a set of relay nodes to enhance the
legitimate users' transmission and another set of relay nodes to
perform jamming of the eavesdroppers. The proposed RJFS and BF-RJFS
selection algorithms can exploit the use of the buffers in the relay
nodes and result in substantial gains in secrecy rate over existing
techniques.

\appendix[Proof of the proposed E-SR criterion]
In this appendix, we include the detailed steps of the derivation of
the E-SR criterion.
\begin{proof}
From the original expressions for the achievable rate for users and
eavesdroppers, which are shown in (\ref{eqn:RR1}) and
(\ref{eqn:RE1}), we consider relay selection based on the secrecy
rate criterion according to: \begin{equation} {\mathcal R}^{\rm SR}
=\arg \max_{\boldsymbol \varphi \in \boldsymbol{\Psi}}\sum_{\phi_i
\in \boldsymbol \varphi} \left\{ \frac{\det(\boldsymbol I
+{\boldsymbol{\Gamma}_{r}^{(t)}})}{\det(\boldsymbol
I+{\boldsymbol{\Gamma}_{e}^{(t)}})}\right\}, \label{eqn:alsinrct4}
\end{equation}
where ${\boldsymbol{\Gamma}}_{r}^{(t)}=({\boldsymbol
H}_{i}{\boldsymbol R}_{I} {\boldsymbol H}_{i}^H)^{-1}({\boldsymbol
H}_{i}{\boldsymbol R}_d {\boldsymbol H}_{i}^H)$ and
${\boldsymbol{\Gamma}}_{e}^{(t)}=({\boldsymbol H}_{e}{\boldsymbol
R}_{I} {\boldsymbol H}_{e}^H)^{-1}({\boldsymbol H}_{e}{\boldsymbol
R}_r {\boldsymbol H}_{e}^H)$. Note that (\ref{eqn:alsinrct4})
requires ${\boldsymbol H}_{e}$, i.e., CSI to the eavesdroppers and
that channels from all relays are taken into account for selection.
In what follows, we show that a designer can employ an equivalent
expression to (\ref{eqn:alsinrct4}) without resorting to the
knowledge of CSI to the eavesdroppers. This requires the assumption
that several channel matrices are square. However, it can also be
used even for scenarios of non-square channel matrices if the
matrices are
completed with zeros to ensure a square structure. 

In (\ref{eqn:alsinrct4}), our aim is to circumvent the need for CSI
to the eavesdroppers from the denominator. To this end, we assume
square matrices which allows the linear algebra property
$\det(\boldsymbol A\boldsymbol B)=\det(\boldsymbol
A)\det(\boldsymbol B)$. Following this approach, the denominator of
(\ref{eqn:alsinrct4}) can be expressed as
\begin{equation}
{\det[{\boldsymbol \Lambda_1}^{-1}{\boldsymbol \Lambda_1}+{\boldsymbol \Lambda_1}^{-1}(\boldsymbol H_{e}\boldsymbol R_{d}\boldsymbol
H_{e}^{H})]}, \label{eqn:alsinrct411}
\end{equation}
where ${\boldsymbol \Lambda_1}=\boldsymbol H_{e}\boldsymbol
R_{I}\boldsymbol H_{e}^{H}$.

Since $\boldsymbol \Lambda_1$ is assumed to be a square matrix,
(\ref{eqn:alsinrct411}) can be decomposed as
\begin{equation}
\det[{\boldsymbol \Lambda_1}^{-1}]{\det[{\boldsymbol \Lambda_1}+(\boldsymbol H_{e}\boldsymbol R_{d}\boldsymbol
H_{e}^{H})]}, \label{eqn:alsinrct412}
\end{equation}
Using the property of the determinant $\det(\boldsymbol
A^{-1})=\frac{1}{\det(\boldsymbol A)}$ \cite{matrix}, we have
\begin{equation}
(\det[{\boldsymbol \Lambda_1}])^{-1}{\det[{\boldsymbol
\Lambda_1}+(\boldsymbol H_{e}\boldsymbol R_{d}\boldsymbol
H_{e}^{H})]}, \label{eqn:alsinrct42}
\end{equation}
where {\small \begin{equation}
\begin{split}
\hspace{-1em}(\det[{\boldsymbol \Lambda_1}])^{-1}& =
\big(\det[\boldsymbol
H_{e}\boldsymbol R_{I}\boldsymbol H_{e}^{H}]\big)^{-1}\\
& = \big(\det[\boldsymbol H_{e}(\sum_{j\neq i}\boldsymbol
U_{j}{\boldsymbol s}_{j}^{(t)}{{\boldsymbol
s}_{j}^{(t)}}^{H}{\boldsymbol U_{j}}^{H})\boldsymbol
H_{e}^{H}]\big)^{-1}\\
& = \big(\det[\boldsymbol H_{e}\underbrace{{\boldsymbol U}_i
{\boldsymbol U}_i^{-1}}_{\boldsymbol I}(\sum_{j\neq i}\boldsymbol
U_{j}{\boldsymbol s}_{j}^{(t)}{{\boldsymbol
s}_{j}^{(t)}}^{H}{\boldsymbol U_{j}}^{H})\underbrace{{\boldsymbol
U}_i^{H^{-1}} {\boldsymbol U}_i^{H}}_{\boldsymbol I}{\boldsymbol
H}_{e}^{H}]\big)^{-1}\\
& = \big(\det[\boldsymbol H_{e}{\boldsymbol
U}_i]\det\big[(\sum_{j\neq i}{\boldsymbol U}_i^{-1}{\boldsymbol
U}_{j}{\boldsymbol s}_{j}^{(t)}{{\boldsymbol
s}_{j}^{(t)}}^{H}{\boldsymbol U_{j}}^{H}{\boldsymbol
U}_i^{H^{-1}})\big] \\
& \det[{\boldsymbol U}_i^{H}{\boldsymbol H}_{e}^{H}]\big)^{-1}
\end{split}
\end{equation}}
and
\begin{equation}
\begin{split}
\det[\underbrace{{\boldsymbol \Lambda_1}+(\boldsymbol
H_{e}\boldsymbol R_{d}{\boldsymbol H}_{e}^{H})}_{\boldsymbol
H_{e}({\boldsymbol R}_{I}+{\boldsymbol R}_{d}){\boldsymbol
H}_{e}^{H}}] & = \det[{\boldsymbol H}_{e}{\boldsymbol U}_i ] \times \\
& \det\big[(\sum_{j\neq i}{\boldsymbol U}_i^{-1}{\boldsymbol
U}_{j}{\boldsymbol s}_{j}^{(t)}{{\boldsymbol s}_{j}^{(t)}}^{H}
{\boldsymbol U_{j}}^{H}{\boldsymbol U}_i^{H^{-1}})
\\ & +{\boldsymbol s}_{i}^{(t)}{{\boldsymbol s}_{i}^{(t)}}^{H}
\big]\det[{\boldsymbol U}_i^{H}{\boldsymbol H}_{e}^{H}]
\end{split}
\end{equation}
The denominator of the argument in (\ref{eqn:alsinrct4}) can then be
expressed as
\begin{equation}
\begin{split}
&(\det[{\boldsymbol \Lambda_1}])^{-1}{\det[{\boldsymbol
\Lambda_1}+(\boldsymbol H_{e}\boldsymbol R_{d}\boldsymbol
H_{e}^{H})]}  =  \\ & \hspace{-2.5em} \big(\det[\boldsymbol
H_{e}{\boldsymbol U}_i]\det\big[(\sum_{j\neq i}{\boldsymbol
U}_i^{-1}{\boldsymbol U}_{j}{\boldsymbol s}_{j}^{(t)}{{\boldsymbol
s}_{j}^{(t)}}^{H}{\boldsymbol U_{j}}^{H}{\boldsymbol
U}_i^{H^{-1}})\big]\det[{\boldsymbol U}_i^{H}{\boldsymbol
H}_{e}^{H}]\big)^{-1} \\ &  \hspace{-1em}\quad \det[{\boldsymbol
H}_{e}{\boldsymbol U}_i ] \det\big[(\sum_{j\neq i}{\boldsymbol
U}_i^{-1}{\boldsymbol U}_{j}{\boldsymbol s}_{j}^{(t)}{{\boldsymbol
s}_{j}^{(t)}}^{H}{\boldsymbol U_{j}}^{H}{\boldsymbol
U}_i^{H^{-1}})+{\boldsymbol s}_{i}^{(t)}{{\boldsymbol
s}_{i}^{(t)}}^{H}  \big]  \\ &   \det[{\boldsymbol
U}_i^{H}{\boldsymbol H}_{e}^{H}] , \label{eqn:alsinrct43}
\end{split}
\end{equation}
Using the matrix inverse property $\big(\det[{\boldsymbol A}]
\det[{\boldsymbol B}] \det[{\boldsymbol C}]\big)^{-1} =
(\det[{\boldsymbol C}^H])^{-1} (\det[{\boldsymbol B}])^{-1}
(\det[{\boldsymbol A}^H])^{-1}$ \cite{matrix}, we can write
\begin{equation}
\begin{split}
\hspace{-1em}&(\det[{\boldsymbol \Lambda_1}])^{-1}{\det[{\boldsymbol
\Lambda_1}+(\boldsymbol H_{e}\boldsymbol R_{d}\boldsymbol
H_{e}^{H})]}  =  \big(\det[{\boldsymbol H}_{e}{\boldsymbol
U}_i]\big)^{-1} \\ & \big(\det\big[(\sum_{j\neq i}{\boldsymbol
U}_i^{-1}{\boldsymbol U}_{j}{\boldsymbol s}_{j}^{(t)}{{\boldsymbol
s}_{j}^{(t)}}^{H}{\boldsymbol U_{j}}^{H}{\boldsymbol
U}_i^{H^{-1}})\big]\big)^{-1} \\  &\hspace{-2em}
\big(\det[{\boldsymbol U}_i^{H}\boldsymbol H_{e}^{H}]\big)^{-1}
\det[{\boldsymbol H}_{e}{\boldsymbol U}_i ] \det\big[(\sum_{j\neq
i}{\boldsymbol U}_i^{-1}{\boldsymbol U}_{j}{\boldsymbol
s}_{j}^{(t)}{{\boldsymbol s}_{j}^{(t)}}^{H}{\boldsymbol
U_{j}}^{H}{\boldsymbol U}_i^{H^{-1}})\\ & \quad +{\boldsymbol
s}_{i}^{(t)}{{\boldsymbol s}_{i}^{(t)}}^{H}  \big]\det[{\boldsymbol
U}_i^{H}{\boldsymbol H}_{e}^{H}] , \label{eqn:alsinrct43}
\end{split}
\end{equation}
By observing the terms above, we notice that that the $1$st term can
be canceled by the $4$th term, and that the $3$rd term can be
canceled by the $6$th term, resulting in
\begin{equation}
\begin{split}
&(\det[{\boldsymbol \Lambda_1}])^{-1}{\det[{\boldsymbol
\Lambda_1}+(\boldsymbol H_{e}\boldsymbol R_{d}\boldsymbol
H_{e}^{H})]}  =  \\ & \big(\det\big[(\sum_{j\neq i}{\boldsymbol
U}_i^{-1}{\boldsymbol U}_{j}{\boldsymbol s}_{j}^{(t)}{{\boldsymbol
s}_{j}^{(t)}}^{H}{\boldsymbol U_{j}}^{H}{\boldsymbol
U}_i^{H^{-1}})\big]\big)^{-1} \\ & \quad \det\big[(\sum_{j\neq
i}{\boldsymbol U}_i^{-1}{\boldsymbol U}_{j}{\boldsymbol
s}_{j}^{(t)}{{\boldsymbol s}_{j}^{(t)}}^{H}{\boldsymbol
U_{j}}^{H}{\boldsymbol U}_i^{H^{-1}})+{\boldsymbol
s}_{i}^{(t)}{{\boldsymbol s}_{i}^{(t)}}^{H}  \big]\\
& = \det[ {\boldsymbol I} + {\boldsymbol U}_i^H {\boldsymbol
R}_{I}^{-1} {\boldsymbol U}_{i}{\boldsymbol R}_d],
\label{eqn:alsinrct44}
\end{split}
\end{equation}
By substituting the result in (\ref{eqn:alsinrct44}) in
(\ref{eqn:alsinrct4}), we obtain the proposed E-SR selection
criterion given by {\small
\begin{equation}
\begin{split}
\hspace{-1em}{\mathcal R}^{\rm S-SR}& =\arg \max_{\boldsymbol
\varphi \in \boldsymbol{\Psi}}\sum_{\phi_i \in \boldsymbol \varphi}
\left\{\log\Big( \frac{\det[\boldsymbol I
+{\boldsymbol{\Gamma}_{r}^{(t)}}]}{\det(\boldsymbol
I+{\boldsymbol{\Gamma}_{e}^{(t)}}]}\Big)\right\}\\
& =\arg \max_{\boldsymbol \varphi \in \boldsymbol{\Psi}}\sum_{\phi_i
\in \boldsymbol \varphi} \left\{\log\Big( \frac{\det[\boldsymbol I
+(\boldsymbol H_{i}\boldsymbol R_{I}\boldsymbol
H_{i}^{H})^{-1}(\boldsymbol H_{i}\boldsymbol R_{d}\boldsymbol
H_{i}^{H})]}{\det[{\boldsymbol I}+({\boldsymbol H}_{e}{\boldsymbol
R}_{I} {\boldsymbol H}_{e}^H)^{-1}({\boldsymbol H}_{e}{\boldsymbol
R}_r {\boldsymbol H}_{e}^H)]}\Big)\right\}\\
& =\rm{arg} \max_{\boldsymbol \varphi \in \boldsymbol{\Psi}}
\sum_{\phi_i \in \boldsymbol \varphi}\left\{\log\Big(
\frac{\det[\boldsymbol I+(\boldsymbol H_{i}\boldsymbol
R_{I}\boldsymbol H_{i}^{H})^{-1}(\boldsymbol H_{i}\boldsymbol
R_{d}\boldsymbol H_{i}^{H})]}{\det[\boldsymbol I + \boldsymbol
U_{i}^{H}\boldsymbol R_{I}^{-1}\boldsymbol U_{i}\boldsymbol
R_{d}]}\Big)\right\}\\
& = \arg \max_{\boldsymbol \varphi \in \boldsymbol{\Psi}}\sum_{i \in
\boldsymbol \varphi} \bigg\{\log\big(\det{\left[\boldsymbol
I+(\boldsymbol H_{i}\boldsymbol R_{I}\boldsymbol
H_{i}^{H})^{-1}(\boldsymbol
H_{i}\boldsymbol R_{d}\boldsymbol H_{i}^{H})\right]}\big) \\
& \qquad \log\big(\det{\left[\boldsymbol I + \boldsymbol
U_{i}^{H}\boldsymbol R_{I}^{-1}\boldsymbol U_{i}\boldsymbol
R_{d}\right]}\big)\bigg\}, \label{eqn:alsinrct451}
\end{split}
\end{equation} }
where the last expression in (\ref{eqn:alsinrct451}) no longer
requires knowledge of CSI to the eavesdroppers ${\boldsymbol H}_{e}$
and is equivalent to (\ref{eqn:alrnew}).

\end{proof}

\ifCLASSOPTIONcaptionsoff
  \newpage
\fi

\bibliographystyle{IEEEtran}
\bibliography{referenceIEEE}

\end{document}